\documentclass[runningheads]{llncs}
\usepackage{amssymb}
\usepackage{mathtools}
\usepackage{xspace}
\usepackage[normalem]{ulem}
\usepackage{enumitem}
\usepackage{upgreek}
\usepackage{tikz}
\usepackage{subcaption}
\usepackage{harpoon}
\usepackage{wasysym}
\newcommand\tiksize{0.12}

\let\epsilon\varepsilon

\let\emptyset\varnothing
\let\leq\leqslant
\let\geq\geqslant
\let\phi\varphi
\let\tau\uptau

\let\bar\overline
\newcommand{\stkout}[1]{\ifmmode\text{\sout{\ensuremath{#1}}}\else\sout{#1}\fi}

\DeclarePairedDelimiter\abs{\lvert}{\rvert}
\DeclarePairedDelimiter{\ceil}{\lceil}{\rceil}
\DeclarePairedDelimiter\floor{\lfloor}{\rfloor}
\DeclarePairedDelimiter\angles{\langle}{\rangle}

\DeclarePairedDelimiter\mset{\{\!\!\{}{\}\!\!\}}

\makeatletter
\newcommand*\circled[2][1.6]{\tikz[baseline=(char.base)]{
    \node[shape=circle, draw, inner sep=1pt,
        minimum height={\f@size*#1},] (char) {\vphantom{WAH1g}#2};}}
\makeatother
\newcommand{\sat}[2]{\ensuremath{#1 \models #2}}

\newcommand{\words}{\Sigma^*}

\newcommand{\wrt}{\text{w.r.t.}\xspace}

\makeatletter
\def\@endtheorem{\qed\endtrivlist\@endpefalse } 
\makeatother
\newcommand{\R}{\mathbb R}%
\newcommand{\Rnn}{\mathbb R_{\ge 0}}%
\newcommand{\Q}{\mathbb Q}%
\newcommand{\N}{\mathbb N}
\newcommand{\B}{\mathbb B}
\newcommand{\cleq}{\preccurlyeq}

\newcommand{\msets}{\mathcal M}
\newcommand{\negc}[1]{k_{#1}}
\newcommand{\posc}[1]{c_{#1}}
\newcommand{\fit}{f}
\newcommand{\agg}{@}

\newcommand{\hDelta}{\widehat\Delta}
\newcommand{\Xsum}[2]{\textit{sum}(#1,#2)}

\newcommand{\valpha}[1]{\widehat{\alpha}_{#1}}
\newcommand{\vbeta}[1]{\widehat{\beta}_{#1}}
\newcommand{\overeq}[1]{\stackrel{#1}{=}}
\DeclareMathOperator*{\avgrate}{@_{\textit{avg}}}

\newcommand{\metaPalt}{\Psi}
\newcommand{\metaMset}{X}

\newcommand{\img}[2]{#1\odot{#2}}
\newcommand{\rest}[2]{#1\!\!\mid_{#2}}

\newcommand{\A}{\ensuremath{\mathbf{A}}}
\newcommand{\D}{\ensuremath{\mathbf{D}}}

\begin{document}

\title{Decoupled Fitness Criteria for Reactive Systems}
%
%

\author{Derek Egolf
\and
Stavros Tripakis
}
\authorrunning{D. Egolf \& S. Tripakis}
%
\institute{
Northeastern University, Boston, MA, USA\\
\email{\{egolf.d, stavros\}@northeastern.edu}}

\maketitle              
\begin{abstract}
The correctness problem for reactive systems has been thoroughly explored and is well understood. Meanwhile, the efficiency problem for reactive systems has not received the same attention. Indeed, one correct system may be less fit than another correct system and determining this manually is challenging and often done ad hoc. We (1) propose a novel and general framework which automatically assigns comparable fitness scores to reactive systems using interpretable parameters that are decoupled from the system being evaluated, (2) state the computational problem of evaluating this fitness score and reduce this problem to a matrix analysis problem, (3) discuss symbolic and numerical methods for solving this matrix analysis problem, and (4) illustrate our approach by evaluating the fitness of nine systems across three case studies, including the Alternating Bit Protocol and Two Phase Commit.
\end{abstract}

\section{Introduction}

Correctness guarantees help us avoid irritating, costly, and, in some cases, deadly implementation bugs. However, two systems
that both
satisfy a
correctness specification may differ with respect to efficiency. Inefficient systems delay content delivery, use excess energy, and waste clock cycles better
spent elsewhere. Any
of these consequences could reduce the sustainability of an institution employing an inefficient system.

Much like reasoning about correctness, reasoning about efficiency is cognitively demanding, prone to errors, and requires expert insight. The framework proposed in this paper strives to eliminate this human burden, mitigate these errors, and capture the expert's insight and intentions in the parameters of the framework.

The proposed framework accomplishes these goals by assigning a comparable
\textit{fitness score} to every system, such that we can decide between two
systems
on the basis of
their score. Consider the following example.

\begin{example}
\label{ex_protocol}

Consider the finite labeled transition systems (LTSs) depicted in Fig.~\ref{fig:protocol}. Labels $s,a,t$ represent \textit{send}, \textit{acknowledge} (ack), and \textit{timeout} respectively.
The symbols !, ? (output, input)
denote rendezvous communication in which a ? transition can only be taken in one LTS if the corresponding ! transition is taken in another LTS. Transitions with neither ?, nor ! can be taken freely.

LTS $E$ represents a {\em sender} in the environment. LTSs $G$ and $B$ are `good' and `bad' receivers, respectively. $B$ is `bad' in the sense that it waits for two \textit{send} actions before replying with an
acknowledgement,
whereas $G$ replies right away. The synchronous products of the sender $E$ with
receivers $G$ and $B$, denoted $E\abs{}G$ and $E\abs{}B$,
are LTSs $M$ and $M'$, respectively.
Both $M$ and $M'$ are {\em correct}, in the sense that
they satisfy the specification
{\em every $s$ is eventually followed by an $a$} (given some fairness assumptions that prevent $a$ from being ignored indefinitely).
Because they both satisfy this specification,  $M$ and $M'$ are indistinguishable from the perspective of traditional verification and synthesis.
However, $M$ is intuitively preferable to $M'$ because $G$ is a better receiver than $B$.
As we will show in Section~\ref{sec_casestudy}, our framework assigns fitness scores $0.25$ and $0.14$ to $M$ and $M'$, respectively, and thus distinguishes $M$ as a better system.
\end{example}

\captionsetup{belowskip=0pt}
\begin{figure*}
    \captionsetup{aboveskip=5pt}
    \begin{subfigure}[b]{0.32\textwidth}
        \centering
\begin{center}
\begin{tikzpicture}[scale=\tiksize]
\tikzstyle{every node}+=[inner sep=0pt]
\draw [black] (20.3,-22.4) circle (3);
\draw (20.3,-22.4) node {$s_0$};
\draw [black] (44.7,-22.4) circle (3);
\draw (44.7,-22.4) node {$s_2$};
\draw [black] (32.5,-22.4) circle (3);
\draw (32.5,-22.4) node {$s_1$};
\draw [black] (22.538,-20.431) arc (120.16829:59.83171:7.685);
\fill [black] (30.26,-20.43) -- (29.82,-19.6) -- (29.32,-20.46);
\draw (26.4,-18.89) node [above] {$s!$};
\draw [black] (34.685,-20.373) arc (121.41673:58.58327:7.512);
\fill [black] (42.52,-20.37) -- (42.09,-19.53) -- (41.57,-20.38);
\draw (38.6,-18.77) node [above] {$t$};
\draw [black] (30.276,-24.385) arc (-59.49707:-120.50293:7.637);
\fill [black] (22.52,-24.39) -- (22.96,-25.22) -- (23.47,-24.36);
\draw (26.4,-25.94) node [below] {$a?$};
\draw [black] (42.454,-24.36) arc (-60.03415:-119.96585:7.715);
\fill [black] (34.75,-24.36) -- (35.19,-25.19) -- (35.69,-24.33);
\draw (38.6,-25.89) node [below] {$s!$};
\draw [black] (15.3,-22.4) -- (17.3,-22.4);
\fill [black] (17.3,-22.4) -- (16.5,-21.9) -- (16.5,-22.9);
\end{tikzpicture}
\end{center}
        \caption[]
        {{The sender $E$}}
        \label{fig:protocol-E}
    \end{subfigure}
    \begin{subfigure}[b]{0.32\textwidth}
                \centering
\begin{center}
\begin{tikzpicture}[scale=\tiksize]
\tikzstyle{every node}+=[inner sep=0pt]
\draw [black] (20.3,-22.4) circle (3);
\draw (20.3,-22.4) node {$g_0$};
\draw [black] (32.5,-22.4) circle (3);
\draw (32.5,-22.4) node {$g_1$};
\draw [black] (22.538,-20.431) arc (120.16829:59.83171:7.685);
\fill [black] (30.26,-20.43) -- (29.82,-19.6) -- (29.32,-20.46);
\draw (26.4,-18.89) node [above] {$s?$};
\draw [black] (30.276,-24.385) arc (-59.49707:-120.50293:7.637);
\fill [black] (22.52,-24.39) -- (22.96,-25.22) -- (23.47,-24.36);
\draw (26.4,-25.94) node [below] {$a!$};
\draw [black] (15.3,-22.4) -- (17.3,-22.4);
\fill [black] (17.3,-22.4) -- (16.5,-21.9) -- (16.5,-22.9);
\draw [black] (35.18,-21.077) arc (144:-144:2.25);
\draw (36.305,-18.75) node [right] {$s?$};
\fill [black] (35.18,-23.72) -- (35.53,-24.6) -- (36.12,-23.79);
\end{tikzpicture}
\end{center}
                \caption[]
                {{A ``good'' receiver $G$}}
                \label{fig:protocol-G}
    \end{subfigure}
    \hfill
    \begin{subfigure}[b]{0.32\textwidth}
                \centering
\begin{center}
\begin{tikzpicture}[scale=\tiksize]
\tikzstyle{every node}+=[inner sep=0pt]
\draw [black] (20.3,-22.4) circle (3);
\draw (20.3,-22.4) node {$b_0$};
\draw [black] (32.5,-22.4) circle (3);
\draw (32.5,-22.4) node {$b_1$};
\draw [black] (44.7,-22.4) circle (3);
\draw (44.7,-22.4) node {$b_2$};
\draw [black] (22.538,-20.431) arc (120.16829:59.83171:7.685);
\fill [black] (30.26,-20.43) -- (29.82,-19.6) -- (29.32,-20.46);
\draw (26.4,-18.89) node [above] {$s?$};
\draw [black] (15.3,-22.4) -- (17.3,-22.4);
\fill [black] (17.3,-22.4) -- (16.5,-21.9) -- (16.5,-22.9);
\draw [black] (34.677,-20.365) arc (121.59519:58.40481:7.488);
\fill [black] (42.52,-20.36) -- (42.1,-19.52) -- (41.58,-20.37);
\draw (38.6,-18.75) node [above] {$s?$};
\draw [black] (42.188,-24.035) arc (-61.22071:-118.77929:20.123);
\fill [black] (22.81,-24.03) -- (23.27,-24.86) -- (23.75,-23.98);
\draw (32.5,-27.02) node [below] {$a!$};
\draw [black] (47.38,-21.077) arc (144:-144:2.25);
\draw (48.505,-18.75) node [right] {$s?$};
\fill [black] (47.38,-23.72) -- (47.73,-24.6) -- (48.32,-23.79);
\end{tikzpicture}
\end{center}
                \caption[]
                {{A ``bad'' receiver $B$}
                }
                \label{fig:protocol-B}
    \end{subfigure}
    \vskip\baselineskip
    \begin{subfigure}[b]{0.475\textwidth}
                \centering

\begin{center}
\begin{tikzpicture}[scale=\tiksize]
\tikzstyle{every node}+=[inner sep=0pt]
\draw [black] (20.3,-22.4) circle (3);
\draw (20.3,-22.4) node {$p_0$};
\draw [black] (44.7,-22.4) circle (3);
\draw (44.7,-22.4) node {$p_2$};
\draw [black] (32.5,-22.4) circle (3);
\draw (32.5,-22.4) node {$p_1$};
\draw [black] (22.538,-20.431) arc (120.16829:59.83171:7.685);
\fill [black] (30.26,-20.43) -- (29.82,-19.6) -- (29.32,-20.46);
\draw (26.4,-18.89) node [above] {$s$};
\draw [black] (34.685,-20.373) arc (121.41673:58.58327:7.512);
\fill [black] (42.52,-20.37) -- (42.09,-19.53) -- (41.57,-20.38);
\draw (38.6,-18.77) node [above] {$t$};
\draw [black] (30.276,-24.385) arc (-59.49707:-120.50293:7.637);
\fill [black] (22.52,-24.39) -- (22.96,-25.22) -- (23.47,-24.36);
\draw (26.4,-25.94) node [below] {$a$};
\draw [black] (42.454,-24.36) arc (-60.03415:-119.96585:7.715);
\fill [black] (34.75,-24.36) -- (35.19,-25.19) -- (35.69,-24.33);
\draw (38.6,-25.89) node [below] {$s$};
\draw [black] (15.3,-22.4) -- (17.3,-22.4);
\fill [black] (17.3,-22.4) -- (16.5,-21.9) -- (16.5,-22.9);
\end{tikzpicture}
\end{center}
                \caption[]
                {{The product system $M := E\abs{}G$}}
                \label{fig:protocol-P}
    \end{subfigure}
    \hfill
    \begin{subfigure}[b]{0.475\textwidth}
                \centering
\begin{center}
\begin{tikzpicture}[scale=\tiksize]
\tikzstyle{every node}+=[inner sep=0pt]
\draw [black] (20.3,-22.4) circle (3);
\draw (20.3,-22.4) node {$p_0'$};
\draw [black] (32.5,-22.4) circle (3);
\draw (32.5,-22.4) node {$p_1'$};
\draw [black] (44.7,-22.4) circle (3);
\draw (44.7,-22.4) node {$p_2'$};
\draw [black] (32.5,-31.5) circle (3);
\draw (32.5,-31.5) node {$p_3'$};
\draw [black] (44.7,-31.5) circle (3);
\draw (44.7,-31.5) node {$p_4'$};
\draw [black] (15.3,-22.4) -- (17.3,-22.4);
\fill [black] (17.3,-22.4) -- (16.5,-21.9) -- (16.5,-22.9);
\draw [black] (23.14,-21.45) arc (102.69278:77.30722:14.835);
\fill [black] (29.66,-21.45) -- (28.99,-20.79) -- (28.77,-21.76);
\draw (26.4,-20.59) node [above] {$s$};
\draw [black] (35.366,-21.528) arc (101.58491:78.41509:16.104);
\fill [black] (41.83,-21.53) -- (41.15,-20.88) -- (40.95,-21.86);
\draw (38.6,-20.7) node [above] {$t$};
\draw [black] (42.3,-24.19) -- (34.9,-29.71);
\fill [black] (34.9,-29.71) -- (35.84,-29.63) -- (35.25,-28.83);
\draw (37.71,-26.45) node [above] {$s$};
\draw [black] (30.1,-29.71) -- (22.7,-24.19);
\fill [black] (22.7,-24.19) -- (23.05,-25.07) -- (23.64,-24.27);
\draw (27.35,-26.45) node [above] {$a$};
\draw [black] (41.82,-32.327) arc (-79.04745:-100.95255:16.949);
\fill [black] (35.38,-32.33) -- (36.07,-32.97) -- (36.26,-31.99);
\draw (38.6,-33.14) node [below] {$s$};
\draw [black] (35.465,-31.051) arc (95.82137:84.17863:30.908);
\fill [black] (41.73,-31.05) -- (40.99,-30.47) -- (40.89,-31.47);
\draw (38.6,-30.39) node [above] {$t$};
\end{tikzpicture}
\end{center}
                \caption[]
                {{The product system $M' := E\abs{}B$}}
                \label{fig:protocol-P'}
    \end{subfigure}
    \captionsetup{belowskip=-12pt}
    \caption[]
    {{A simple communication protocol modeled with finite LTSs.}}
    \label{fig:protocol}
\end{figure*}
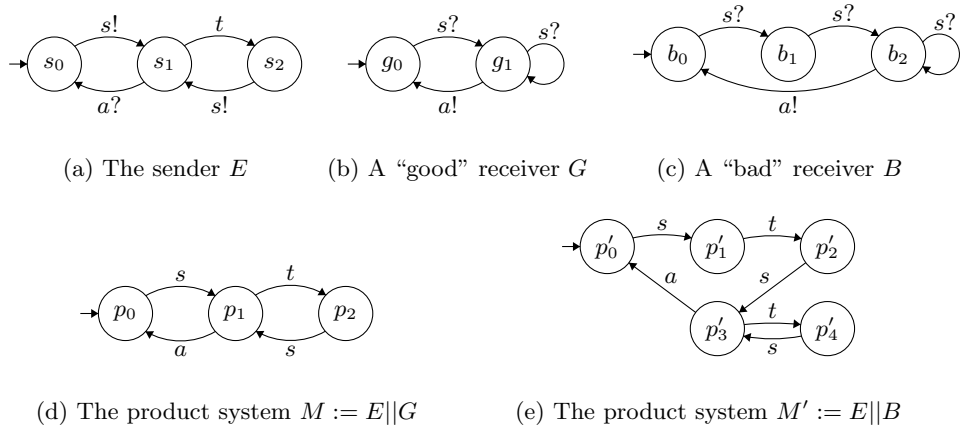
\captionsetup{belowskip=-15pt}

The exact nature of the fitness score depends on the application domain.
Our framework {\em decouples} the description of the system (e.g., the LTSs of Fig.~\ref{fig:protocol}) from a set of domain-specific {\em parameters} which capture user preferences.

By assigning fitness scores to systems, as in the example above, our framework can be used for performance evaluation.
Our framework is additionally motivated by recent work in the synthesis of distributed protocols~\cite{ScenariosHVC2014}.
Unlike humans, synthesis tools typically ignore efficiency considerations.
In some cases, these tools generate systems that are, strictly speaking, correct (i.e., they satisfy their logical specification), yet clearly unorthodox or even inefficient~\cite{AlurTripakisSIGACT17}.
In such cases, we can use our framework to rank automatically generated systems according to their fitness score.
In other cases, we may want to generate {\em all} correct systems~\cite{ATVA2023}, potentially with the aim of
doing {\em fitness-optimal synthesis} (c.f.~Appendix~\ref{sec_fitnessoptimalsynthesis}.).

In summary, the contributions of this paper are as follows:
(1) We propose a novel and general framework for automatically assigning a comparable fitness score to a system; this framework uses interpretable parameters that are decoupled from the system being evaluated.
(2) We provide an automated method for computing fitness scores; our method ultimately reduces the fitness-score computation problem to a matrix analysis problem.
(3) We discuss symbolic and numerical methods for solving this matrix analysis problem.
(4) We present an implementation and evaluation of our framework: our prototype tool allows, in a matter of seconds, to automatically compute the fitness of nine automatically synthesized systems.

We organize the rest of the paper as follows.
Section~\ref{sec_prelim} formalizes preliminary concepts.
Section~\ref{sec_framework} presents our framework, both in its full (semantical) generality and also as a finitely representable instance that can be treated algorithmically.
Section~\ref{sec:algo} presents a method to compute fitness scores for an instance of our framework.
Section~\ref{sec_casestudy} illustrates our approach on the communication protocol of Example~\ref{ex_protocol}, Two Phase Commit, and the Alternating Bit Protocol taken from~\cite{AlurTripakisSIGACT17}.
Section~\ref{sec_related} discusses related work.
Section~\ref{sec_concl} concludes the paper.

\section{Preliminaries}
\label{sec_prelim}

$\N$, $\Q$, $\R$, $\Rnn$, and $\B$ denote the sets of naturals, rationals, reals, non-negative reals, and booleans, respectively.
A function $h:\N^d\to\Q^{d'}$ is a {\em scalar arithmetic function} if $h$ can be written in terms of basic scalar arithmetic operations $+, -, \times, /$, applied to its natural number arguments.

We often formalize the semantics of a system, $M$, and its specifications, $\phi$
as subsets of $\Sigma^\omega$. When verifying that $M$ satisfies $\phi$, i.e., $M\subseteq\phi$, we do not usually consider the relative abundance of traces produced by the operational definition of $M$. We need only show that $\tau\in M$ implies that $\tau\in\phi$. In that paradigm, we disregard that there may be many ways to generate $\tau$ using $M$.

Our framework for measuring performance does not disregard the relative abundance of traces.
All else equal, if a system is capable of producing the same `unfit' trace by executing any one of many distinct runs, then that system is worse than a system that can produce the unfit trace in just one particular way.
Also, we might consider aggregates like average, mode, sum, standard deviation, etc.
and these all depend on the multiplicity of elements. Toward preserving multiplicity, we define our notation for
multisets. We also define a denotational formulation of systems that does not
abstract away the relative abundance of traces.

\subsubsection{Multisets}
A {\em multiset} $\metaMset$ over domain $D$ is a function $\metaMset:D\to\N$,
where $\metaMset(x)$ represents the {\em multiplicity} of element $x$, i.e.,
how many times $x$ occurs in $\metaMset$.
$\msets(D)$ denotes the class of all multisets over $D$, i.e., the set of all functions $\metaMset:D\to\N$.
If $\metaMset(x) = m$, then we write $x\in_m \metaMset$ (possibly, $m=0$).
The {\em cardinality} of $\metaMset$, denoted $\abs{X}$, is the sum of the multiplicities of all members of the domain $D$.
We write multisets as $\mset{...}$ to differentiate them from sets.

\begin{example}
	We denote by $X = \mset{0,0,1,1,1}$ the multiset where
     $0\in_2 X$ and $1\in_3 X$. Then:
    $\abs{X} = 2 + 3 = 5$.
\end{example}

If $A\subseteq D$ and $X:D\to\N$ is a multiset, then $X$ {\em restricted to} $A$ is a new multiset, denoted $\rest{X}{A}:D\to\N$ and defined as follows.
If $x\notin A$, $\rest{X}{A}(x) = 0$
and otherwise if $x\in A$, then $\rest{X}{A}(x) = X(x)$.
Let $\metaMset:D\to\N$ be a multiset and let $f:D\to D'$ be a function.
Then intuitively, the {\em image of $\metaMset$ by $f$} is a multiset denoted $\img{f}{\metaMset}$ obtained by applying $f$ to the members of $\metaMset$.
E.g. if $f(x) = x^2$, then $\img{f}{\mset{2,-2,3,3,3}} = \mset{4,4,9,9,9}$. Formally, we define $\img{f}{\metaMset}:D'\to\N$ as follows.
$(\img{f}{X})(y) := \abs{(\rest{X}{D_y})}$, where $D_y := \{x\in D\mid f(x) = y\}$.
We may treat a set as a multiset with all multiplicities as 0 or 1 and take its image by $f$ to obtain a multiset.
If $X\in\msets(\N^d)$ and $1\leq i\leq d$, then $\Xsum{X}{i} = \sum_{x\in_c X} cx_i$, where $x_i$ is the $i$th component of $x\in\N^d$.
E.g. $\Xsum{\mset{(1,2), (1,2), (3,4)}}{2} = 2 + 2 + 4$.

\subsubsection{Systems}
As mentioned, the semantic formulation of systems as subsets abstracts away structure that our framework needs.
Namely, for a system $M$ and a finite trace $\pi\in\Sigma^n$, there may be
many partial runs through system $M$ which produce $\pi$. Our framework
requires this structure, so we give an alternate semantic formulation of
systems.

\begin{definition}[Abstract Denotation]
    We characterize a system $M$ by an infinite family of multisets, $(M_n)_{n\in\N}$. For all $i\in\N$, $M_i\in\msets(\Sigma^i)$---i.e. the multiset indexed by $i$ assigns a multiplicity to all finite traces of length $i$.
\end{definition}

This characterization of systems abstracts away the notion of
states while maintaining the multiplicity of finite prefixes. However, there are restrictions on
which families of multisets characterize well-defined systems;
those restrictions follow.

\begin{definition}
    A family of multisets $(M_n)_{n\in\N}$ denotes a system if and
    only if for all $\pi,\pi'\in\Sigma^*$ such that $\pi$ is a prefix of $\pi'$,
    $\pi\in_0 M_{\abs{\pi}} \implies \pi'\in_0 M_{\abs{\pi'}}$
\end{definition}

This restriction enforces the following intuition: for a partial run to produce $\pi'$, there must be at least one partial run for each prefix of $\pi'$.
No further restrictions are necessary. A partial run that produces a prefix of $\pi'$ may either have zero or many continuations which produce $\pi'$.

\begin{example}[Two Systems]
    \label{ex:two-systems}
    We now define two systems $M^{(1)}$ and $M^{(2)}$, each as a family of multisets of finite prefixes over alphabet $\Sigma = \{0,\$\}$.
    The user may interpret these traces
    as follows: \$'s are money that we receive, and 0's are lapses in this income. Intuitively, we prefer behaviors that maximize the rate at which we receive \$'s.

    All partial runs of $M^{(1)}$ produce just one finite trace and this finite trace has multiplicity 1. In particular, $\$^n\in_1 M^{(1)}_n$ and if $w\neq\$^n$, then $w\in_0M^{(1)}_n$. In other words, $M^{(1)}$ is the system that generates prefixes of $\$^\omega$, each with multiplicity 1. We may simply express this system as
    $M^{(1)}_n := \mset{\$^n}$.

    Similarly, the partial runs of $M^{(2)}$ also produce just one finite trace each with multiplicity 1. All such finite traces are alternating \$ and 0, i.e. the sole partial trace of length $n$ is the $n$-length prefix of $(\$0)^\omega$. This system may be expressed as
    $M^{(2)}_n := \mset{(\$0)^{\floor{n/2}}\$^{(n\bmod 2)}}$,
    i.e. even length prefixes end in 0 and odd length prefixes end in \$.
\end{example}

\section{Formal Framework}
\label{sec_framework}

The framework assigns a comparable value called a {\it fitness score} to every system. In this section, we define
this score formally.
We first present the general, semantical framework (Section~\ref{sec_general}).
We then instantiate this general framework and state the main problem solved in this paper (Section~\ref{sec_finite_instantiation}).

\subsection{The General Framework}
\label{sec_general}

The key idea of our framework is that it {\em decouples} the description of the system from the following set of domain-specific framework parameters:
(1) A finite alphabet $\Sigma$, e.g., $\{0,\$\}$.
(2) A {\it{fitness function}}, $\fit:\Sigma^*\to\N^d$. This function measures finite prefixes of infinite traces.
\label{aggregate_type}
(3) An {\it{aggregate function}}, $\agg:\msets(\N^d)\to\Q^{d'}$. This function takes a multiset of fitness values and compiles them into a single value. Examples include min, max, average, etc. taken over arithmetic combinations of natural numbers. The dimensionality of the output, $d'$, enables lexicographic aggregates.
Given these two functions, the framework assigns a {\it fitness score} to every system.
The fitness score is a $d'$-dimensional vector, defined formally in Definition~\ref{def_fitnessscore}.
In addition, the framework may also include:
(4) A comparison relation, $\cleq$,
used to compare the fitness scores of two different systems (c.f. eq.~(\ref{def:cleq-form})).
We next provide examples of the above concepts.

\subsubsection{Fitness Function:}

The {\it rate} function is an example of a fitness function:
    \begin{definition}[Fitness Function: Rate of \$]
        \label{def:income-rate}
        For $\Sigma = \{0,\$\}$ define $\textit{rate$_\$$}(w) = (\#_\$(w), \abs{w})$, where $\#_\$(w)$ is the number of \$'s in $w$ and $\abs{w}$ is the length of $w$.
    \end{definition}
    \begin{example}[Rate of \$ Applied]
        \label{ex:income-rate}
        Recall the systems $M^{(1)}_n = \mset{\$^n}$ and $M^{(2)}_n = \mset{(\$0)^{\floor{n/2}}\$^{(n\bmod 2)}}$ from Example \ref{ex:two-systems}. We apply $f := \textit{rate$_\$$}$ to the $n$-length partial runs of these systems. Taking the image of $M^{(1)}$ and $M^{(2)}$ by $f$ yields:
        \begin{align*}
            &\img{f}{M^{(1)}_n} = \mset{f(\$^n)} = \mset{(n,n)}
            \\
            &\img{f}{M^{(2)}_n} = \mset{f((\$0)^{\floor{n/2}}\$^{(n\bmod 2)})} = \mset{(\ceil{n/2},n)}
        \end{align*}
    \end{example}

\subsubsection{Aggregate Functions:}
The {\it average rate} function is one example of an aggregate function. It
     treats ordered pairs as fractions and takes the average value:
    \begin{definition}[Aggregate Function: Average Rate]
        \label{def:avg-rate}
        For $X\in\msets(\N^2)$, let:
        $$\avgrate(X) = \frac{1}{\abs{X}}\sum_{(p,q)\in_m X}m\cdot\frac{p}{q}$$
    \end{definition}
    \begin{example}
        This example emphasizes the role of multiplicity in aggregates. For instance, if $X := \{(1,3),(1,3),(2,3)\}$,
        then the (1,3) term is counted twice:
        \begin{align*}
            \avgrate(X)
            \ \ =\ \ \frac{1}{\abs{X}}\sum_{(p,q)\in_m X}m\cdot\frac{p}{q}
            \ \ =\ \ \frac{1}{3}(2\cdot\frac{1}{3} + \frac{2}{3})
            \ \ =\ \ 4/9
        \end{align*}
    \end{example}
    \begin{example}
        \label{ex:avg-rate}
        This example applies $\avgrate$ to the running example (Example~\ref{ex:two-systems}). The average is moot here as there is only one partial trace of each length. Recall from Example \ref{ex:income-rate} that $\img{f}{M^{(1)}_n} = \mset{(n,n)}$ and $\img{f}{M^{(2)}_n} = \mset{(\ceil{n/2}, n)}$,
        where $f := \textit{rate$_\$$}$. We can apply average rate to these images:
        $\avgrate(\img{f}{M^{(1)}_n}) = n/n = 1$
        and
        $\avgrate(\img{f}{M^{(2)}_n}) = \ceil{n/2}/n$.
    \end{example}

Another example of an aggregate function is the {\it maximum rate} function:
    \begin{definition}[Aggregate Function: Maximum Rate]
        \label{def:max-rate}
        For $X\in\msets(\N^2)$:
        $$\agg(X) = \max\{p/q\mid (p,q)\in X\}$$
    \end{definition}
    \begin{example}
        For instance, if $X := \mset{(1,3),(1,3),(2,3)}$, then:
        $$\agg(X) = \max\{1/3, 1/3, 2/3\} = 2/3$$
    \end{example}

Another example of an aggregate function is the {\it lexicographic} function:
    \begin{definition}[Aggregate Function: Lexicographic]
        For $X\in\msets(\N^2)$. Here is an example where $d' = 2 \neq 1$.
        Let $\agg_1(X)$ be Average Rate from Definition \ref{def:avg-rate} and let $\agg_2(X)$ be Maximum Rate from Definition \ref{def:max-rate}. Finally define
        $$\agg(X) = (\textstyle\agg_1(X), \textstyle\agg_2(X))$$
        Here, the first component of the aggregate is the average rate and the second component is the maximum rate.
    \end{definition}
    \begin{example}
        \label{ex:lex-rate}
        For $X_1 = \mset{(3,6), (2,6), (4,6)}$ and $X_2 = \mset{(1,4),(2,4),(3,4)}$:
        $$\agg(X_1) = (\textstyle\agg_1(X_1), \textstyle\agg_2(X_1)) = (1/2, 2/3)$$
        $$\agg(X_2) = (\textstyle\agg_1(X_2), \textstyle\agg_2(X_2)) = (1/2, 3/4)$$
        This sort of lexicographic aggregate (together
        with a corresponding comparison relation, see below) can
        be useful for breaking
        ties between choices like $X_1$ vs. $X_2$.
That is, we may, for instance, want to maximize the average rate when possible, but if two distinct choices yield the same average rate, we may want to make the choice that has the highest potential payoff; in that case, we prefer $X_2$.
    \end{example}

\subsubsection{Fitness Score:}

Given the above parameters, our framework assigns a fitness score to every system $M = (M_n)_{n\in\N}$. It does so as follows:
\begin{definition}[Fitness score]
    \label{def_fitnessscore}
    The fitness score of system $M$ is
    $$\agg_f M := \lim_{n\to\infty}\agg(\img{f}{M_n})$$
    This limit is a value in $(\Rnn\cup\{\infty, \bot\})^{d'}$.
    Each component of the vector:
	either converges to a value $v\in\Rnn$, in which case we assign the component the value $v$;
	or increases without bound, in which case we assign the value $\infty$;
	or exhibits some other behavior such as oscillation, in which case we assign the ill-behaved value $\bot$.
\end{definition}
	The framework is quite general, so there are systems that have oscillating fitness scores, i.e., $\bot$ (see Example \ref{ex:osc-fitness} in the Appendix). In what follows, we limit our attention specifically to the class of systems representable by finite transition systems.
	Whether this class contains systems that have $\bot$ fitness scores remains open. However, the system in Example~\ref{ex:osc-fitness} does not belong to this class.

\subsubsection{Comparison Relations:}

    A comparison relation $\cleq$ is a subset of
    \begin{equation}
        \label{def:cleq-form}
    (\Rnn\cup\{\infty, \bot\})^{d'}\times(\Rnn\cup\{\infty, \bot\})^{d'}
    \end{equation}
    If $(a,b)\in\ \cleq$, we write $a\cleq b$. If neither $a\cleq b$ nor $b\cleq a$, we say that $a$ and $b$ are {\em incomparable}.

    Ignoring $\infty$ and $\bot$ for the moment, $\cleq$ could be the relation $\leq$ on $\R$ when $d'=1$, or the  lexicographic comparator when $d'=2$: $$(a_1,b_1)\cleq(a_2,b_2)\iff(a_1 > a_2) \vee (a_1 = a_2 \wedge b_1\geq b_2)$$
This comparator is the sort we would want for Example \ref{ex:lex-rate}.

Extending the above to $\infty$ and $\bot$ would be up to the user. One choice is to have these values be incomparable to any other value.
We also remark that $\cleq$ needs to compare real (and not just rational) numbers, even though the aggregate function $\agg$ maps to $\Q^{d'}$, because the fitness score involves taking a limit.
The semantics of $a\cleq b$ are that $a$ is preferrable to $b$.

\begin{example}
Concluding our analysis of Example~\ref{ex:two-systems},
consider an instance of our framework with fitness function \textit{rate$_\$$} (Definition~\ref{def:income-rate}), aggregate function $\avgrate$ (Definition~\ref{def:avg-rate}), and comparison operator $\cleq\ :=\ \geq$ (since we prefer high rates of income). We can then compare the two simple systems introduced in Example~\ref{ex:two-systems}. Building on what we have presented so far (c.f. Examples~\ref{ex:income-rate} and \ref{ex:avg-rate}), we have:
    \begin{align*}
        \agg_f M^{(1)}
        &
        \ \ =\ \  \lim_{n\to\infty}\agg(\img{f}{M^{(1)}_n})
        \ \ =\ \  \lim_{n\to\infty} 1
        \ \ =\ \  1
        \\
        \agg_f M^{(2)}
        &
        \ \ =\ \  \lim_{n\to\infty}\agg(\img{f}{M^{(2)}_n})
        \ \ =\ \  \lim_{n\to\infty}\frac{\ceil{n/2}}{n}
        \ \ =\ \  1/2
    \end{align*}
    Because $\agg_f M^{(1)} \geq \agg_f M^{(2)}$, we conclude $\agg_f M^{(1)} \cleq \agg_f M^{(2)}$ and
    therefore we prefer $M^{(1)}$ to $M^{(2)}$. This result aligns with our
    intuitions; we would rather receive a dollar every day than a dollar every
    other day.
\end{example}

\subsubsection{Evaluation, Comparison, and Synthesis Problems:}
Within our framework, we can consider various types of computational problems.
A basic problem is that of {\it evaluating} the fitness score of a given system:
{\it Given a fitness function $f$, an aggregate function $\agg$, and a system $M$, compute $\agg_f M$.}
Another problem is that of {\it comparing} two systems:
{\it Given a fitness function $f$, an aggregate function $\agg$, a comparison relation $\cleq$, and two systems $M_1,M_2$,
check whether $\agg_f M_1 \cleq \agg_f M_2$.}
We can also consider {\it fitness-optimal synthesis} problems like the one presented in Appendix~\ref{sec_fitnessoptimalsynthesis}.

The problems described above are abstract in the sense that our framework is semantical.
In order to
define concrete computational problems of this sort, we need some concrete, syntactic representation of the elements of our framework, namely, systems, fitness functions, etc.
We present one such representation in Section~\ref{sec_finite_instantiation} that follows.

\subsection{Syntactic Representation of the Framework}
\label{sec_finite_instantiation}

We represent systems using {\it finite labeled transition systems}:
\begin{definition}[Finite Labeled Transition System]
	A finite labeled transition system (LTS) is a tuple $M = \angles{\Sigma, Q, Q_0, \Delta}$, where
    \begin{itemize}
        \item $\Sigma$ is a finite set of labels
        \item $Q$ is a finite set of states
        \item $Q_0\subseteq Q$ is the set of initial states
        \item $\Delta\subseteq Q\times \Sigma\times Q$ is a transition
        relation
    \end{itemize}
\end{definition}
We now define the denotation, $(M_n)_{n\in\N}$, of labeled transition
system $M$. We first define the {\it{path relation}} of $M$,
    $\hDelta\subseteq Q^*\times\Sigma^*$, in terms of its members.
    Let $\hat q = q^0,...,q^{k}\in
    Q^*$ and $w=w_1w_2...w_{k}\in\Sigma^*$. Then $(\hat q,w)\in\hDelta$
    if and only if:
        (1) $q^0 \in Q_0$ and
        (2) For all $i$ such that $0\leq i < k$,
        $(q^i,w_{i+1},q^{i+1})\in\Delta$.
Note: $(q_0,\epsilon)\in\hDelta$ where $q_0\in Q_0$ and $\epsilon$ is the empty sequence of labels.

Then we define $(M_n)_{n\in\N}$ by
    defining each $M_n$.
    We define $M_n$ by defining the multiplicity
    of each $w\in\Sigma^n$.
    Namely, for a fixed $w\in\Sigma^n$, $w\in_c M_n$ if
    and only if
    $$c = \abs{\{\hat q\in Q^*\mid(\hat q,w)\in\hDelta\}}$$

We represent fitness functions by {\it deterministic finite state automata} (DFA).
Specifically,
    a fitness function $f:\Sigma^*\to\N^d$ is represented by a $d$-tuple
    $\angles{f_1,...,f_d}$, where each $f_i$ is a DFA defined as follows:

\begin{definition}[DFA]
 A DFA is a tuple
    $f_i = \angles{\Sigma,Q_i,q_i^0, Q_i^{\textit{acc}},\delta_i}$,
    where
    \begin{itemize}
        \item $\Sigma$ is a finite set of labels
        \item $Q_i$ is a finite set of states
        \item $q_i^0\in Q_i$ is the single initial state of the automaton
        \item $Q_i^{\textit{acc}}\subseteq Q_i$ is the set of accepting states
        \item $\delta_i:Q_i\times\Sigma\to Q_i$ is the transition function
    \end{itemize}
\end{definition}
Now consider an input $w\in\words$.
When
the DFA $f_i$
consumes $w$, it visits a sequence of states, $\hat{q}=q_i^0, q_i^1, ...,q_i^m$.
Interpreting $f_i$ as a function $f_i : \Sigma^*\to\N$,
	we define $f_i(w)$ as the number of times an accepting state is
    visited in $\hat{q}$.
	We then define the fitness function $f:\Sigma^*\to\N^d$ so that $f(w) = (f_1(w),...,f_d(w))$.
    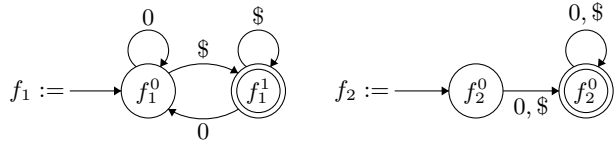
\begin{figure}[h]
        \centering
\begin{center}
\begin{tikzpicture}[scale=\tiksize]
\tikzstyle{every node}+=[inner sep=0pt]
\draw [black] (51.6,-29.2) circle (3);
\draw (51.6,-29.2) node {$f_2^0$};
\draw [black] (15.3,-29.2) circle (3);
\draw (15.3,-29.2) node {$f_1^0$};
\draw [black] (27.5,-29.2) circle (3);
\draw (27.5,-29.2) node {$f_1^1$};
\draw [black] (27.5,-29.2) circle (2.4);
\draw [black] (63.8,-29.2) circle (3);
\draw (63.8,-29.2) node {$f_2^0$};
\draw [black] (63.8,-29.2) circle (2.4);
\draw [black] (6.6,-29.2) -- (12.3,-29.2);
\draw (6.1,-29.2) node [left] {$f_1:=$};
\fill [black] (12.3,-29.2) -- (11.5,-28.7) -- (11.5,-29.7);
\draw [black] (42.6,-29.2) -- (48.6,-29.2);
\draw (42.1,-29.2) node [left] {$f_2:=$};
\fill [black] (48.6,-29.2) -- (47.8,-28.7) -- (47.8,-29.7);
\draw [black] (13.977,-26.52) arc (234:-54:2.25);
\draw (15.3,-21.95) node [above] {$0$};
\fill [black] (16.62,-26.52) -- (17.5,-26.17) -- (16.69,-25.58);
\draw [black] (26.177,-26.52) arc (234:-54:2.25);
\draw (27.5,-21.95) node [above] {$\$$};
\fill [black] (28.82,-26.52) -- (29.7,-26.17) -- (28.89,-25.58);
\draw [black] (62.477,-26.52) arc (234:-54:2.25);
\draw (63.8,-21.95) node [above] {$0,\$$};
\fill [black] (65.12,-26.52) -- (66,-26.17) -- (65.19,-25.58);
\draw [black] (17.538,-27.231) arc (120.15965:59.84035:7.687);
\fill [black] (25.26,-27.23) -- (24.82,-26.4) -- (24.32,-27.26);
\draw (21.4,-25.69) node [above] {$\$$};
\draw [black] (25.303,-31.213) arc (-58.88262:-121.11738:7.552);
\fill [black] (17.5,-31.21) -- (17.92,-32.05) -- (18.44,-31.2);
\draw (21.4,-32.8) node [below] {$0$};
\draw [black] (54.6,-29.2) -- (60.8,-29.2);
\fill [black] (60.8,-29.2) -- (60,-28.7) -- (60,-29.7);
\draw (57.7,-29.7) node [below] {$0,\$$};
\end{tikzpicture}
\end{center}
        \caption{Two examples of DFA representing fitness functions: $f_1$ computes the number of $\$$'s in a word;
        $f_2$ computes the length of the word.}
        \label{fig:fit-repr}
    \end{figure}
\begin{example}[Rate]
\label{ex_rate_DFA}
    Let $f_1(w) :=$ the number of $\$$'s in $w$ and $f_2(w) :=$ the length of $w$.
    We represent these individual components of $f = \angles{f_1, f_2}$ by
    the DFA in Fig.~\ref{fig:fit-repr}. The example is detailed further in
	Appendix~\ref{sec_example_exrateDFA_continued}.
\end{example}

In principle, an aggregate function can be any mathematical function with the appropriate type (c.f. page \pageref{aggregate_type}).
But for the sake of computation, we want an aggregate function to be represented as a scalar arithmetic function $h(x_1,x_2,...,x_d)$.
We say that $h : \N^d\to\Q^{d'}$ is a {\em faithful representation} of $\agg : \msets(N^d)\to\Q^{d'}$ if and only if for all $X\in\msets(\N^d), \agg(X) = h(\Xsum{X}{1},...,\Xsum{X}{d})$.
We will see in Section~\ref{sec:algo} that this form of representation and the definitions that follow are key, as the heart of our method is computing each $\Xsum{X}{i}$, where $X = \img{f}{M_n}$. The importance should be clear by the time we state our primary correctness result, Theorem~\ref{thm:main-correctness}.

While $h$ might not be a faithful representation of $\agg$ for all $X$, $h$ may be a faithful representation assuming that $X$ satisfies some condition. The fitness function may in turn guarantee that $X$ satisfies that condition.
Fortunately, this relationship holds between $\avgrate$ (Def.~\ref{def:avg-rate}) and $\textit{rate}_{\$}$ (Def.~\ref{def:income-rate}).
The following definition and lemmas capture this useful situation:

\begin{definition}[Conditional Representation and Compatible]
\label{defcondcomposep}
    Let $\metaPalt$ be a predicate over $\msets(\N^d)$, i.e., a mapping $\metaPalt:\msets(\N^d)\to\B$.
    Additionally, let $\agg:\msets(\N^d)\to\Q^{d'}$ be an aggregate function and $h : \N^d\to\Q^{d'}$ be a scalar arithmetic function.
    Then $h$ is a {\em conditional representation}
    of $\agg$ subject to $\metaPalt$
    if and only if
    for all $X\in\msets(\N^d)$, if $\metaPalt(X)$ holds (i.e., $\metaPalt(X)=1$), then
        $\agg(X) = h(\Xsum{X}{1}, ..., \Xsum{X}{d})$.

	Let $h$ be a conditional representation of the aggregate function $\agg$ subject to $\metaPalt$. Let $f$ be a fitness function.
	We say that $h$ and $f$ are {\em compatible}
	when $\metaPalt(\img{f}{M_n})$ holds for any LTS $M$ and any $n\in\N$.
\end{definition}

Let predicate  $\metaPalt_\textit{rate}(X) := \text{`If $(p,q), (p',q')\in X$, then $q = q'$.'}$ Then we have the following two lemmas.
\begin{lemma}
\label{thm_avg_agg_cond_comp_separable}
    Let $X\in \msets(\N^2)$ and suppose $\metaPalt_\textit{rate}(X)$ holds. Then
    $\avgrate(X)
    = \Xsum{X}{1}/\Xsum{X}{2}
    $.
     Therefore, $\avgrate$ is conditionally represented by $h(x_1,x_2) = x_1/x_2$, subject to
    $\metaPalt_\textit{rate}$.
\end{lemma}
\begin{lemma}
\label{thm_rate_sat_pred}
    For all $n\in\N$ and all LTS $M$, $\metaPalt_\textit{rate}(\img{\textit{rate}_\$}{M_n})$ holds.
    Hence, $\textit{rate}_\$$ and $h(x_1,x_2) = x_1/x_2$ are compatible.
\end{lemma}

\noindent
Lemma~\ref{thm_avg_agg_cond_comp_separable} follows from the fact that the average of a multiset of fractions is equal to the sum of the numerators divided by the sum of the denominators when the denominators are all equal.
Lemma~\ref{thm_rate_sat_pred} is immediate: if $w\in M_n$ and $\textit{rate}_\$(w) = (p,q)$, then $q = n$.
From Lemma~\ref{thm_avg_agg_cond_comp_separable} and~\ref{thm_rate_sat_pred} it follows that $\avgrate$ and $\textit{rate}_\$$ are compatible.
Therefore,
if the fitness function is $\textit{rate}_\$$
we can represent $\avgrate(X)$ with the expression $\Xsum{X}{1} / \Xsum{X}{2}$.

Note that fitness functions other than $\textit{rate}_\$$ might not
be compatible with $\avgrate$.
 For instance, let $f(w) = (\#_\$(w), \#_0(w))$, which measures the number of \$'s per 0.
$f$ does not satisfy $\metaPalt_\textit{rate}$, but it is a realistic fitness function.
In the case of $\textit{rate}_\$$,  time is measured by the observation of any label from $\Sigma$. Now for $f$, time is measured using only 0. If $\$$ denotes a local action of a server and $0$ an interaction between two servers, $f$ captures communication complexity. We leave handling of such non-compatible fitness functions for future work.

\subsubsection{The Fitness Evaluation Problem:}

We are now ready to state the fitness-score evaluation problem for systems represented as finite LTSs,
fitness functions represented as DFA, and aggregate functions represented as arithmetic expressions.
We provide a solution to this problem
in Section~\ref{sec:algo}.

\begin{problem}[Fitness Evaluation Problem]
\label{prob_}
Let $M = \angles{\Sigma,Q,Q_0,\Delta}$ be a finite LTS and let
$f = \angles{f_1,..., f_d}$, where each $f_i$
is represented as a DFA.
Let $\agg:\msets(\N^d)\to\Q^{d'}$ be an aggregate function
represented by the scalar arithmetic function $h : \N^d\to\Q^{d'}$. Finally, suppose that $h$ and $f$ are compatible.
The fitness evaluation problem is to compute the fitness score $\agg_f M$ of $M$, i.e.,
to compute $\lim_{n\to\infty}\agg(\img{f}{M_n})$.
\end{problem}

\section{Reducing Fitness Evaluation to Matrix Analysis}
\label{sec:algo}

In this section we propose a method to solve Problem~\ref{prob_} that consists in the following steps
(assuming the same notation and setup as in Problem~\ref{prob_}):
\begin{enumerate}
    \item \label{step_product} Compute the product automaton $P_i = M\abs{}f_i$, for each $i\in\{1,...,d\}$.
    \item \label{step_recur} For each $P_i$, compute a matrix-vector pair ($\xi_i$,$v_i$)
    representing a {\em recurrence relation}. We call the matrix $\xi_i$ the {\em recurrence matrix} and the vector $v_i$ the {\em initial condition vector}.
    \item \label{step_matrix}
	Solve the following matrix analysis problem:
\end{enumerate}

\begin{problem}
\label{prob2}
    Let $g_i(n) = (\xi_i^{n+1} v_i)_0$ for fixed square matrices $\xi_1,..., \xi_d$ and vectors $v_1,..., v_d$ with non-negative integer entries and where $(u)_0$ denotes the first entry of vector $u$. Let $h:\N^d\to\Q^{d'}$ be
    a scalar arithmetic function.
    Compute $$\lim_{n\to\infty}h(g_1(n), g_2(n),..., g_d(n))$$
\end{problem}

The motivation for the above steps follows. In step 1, the product $P_i$ represented all simultaneous paths through $M$ and $f_i$. I.e., a path through $P_i$ corresponds to taking a path through $M$ and handing the transition label encountered at each step to the automaton representing $f_i$. As mentioned, step 2 computes a recurrence relation, which is reasonable because the number of accepting states visited across $(n+1)$-length paths is related to certain quantities computed over the $n$-length paths. The exact relationship is explained in detail in Section~\ref{sec_recur_rel}.

The correctness of the reduction to Problem~\ref{prob2}
(Corollary~\ref{cor:main-correctness})
hinges on the fact that $g_i(n) = \Xsum{\img{f}{M_n}}{i}$, i.e., computing $\Xsum{\img{f}{M_n}}{i}$ (which is then an input to the aggregate function) reduces to computing the $n$th term of a recurrence relation, which in turn reduces to taking a matrix power.

Step~\ref{step_product} of the method (computing automata products) is standard.
Therefore, in the rest of this section, we focus on explaining Steps~\ref{step_recur} and~\ref{step_matrix}.

\subsection{Step~\ref{step_recur}: Constructing the Recurrence Relation}
\label{sec_recur_rel}

We will first explain the recurrence relation construction by example and then give the general construction.

\subsubsection{By example:}
 We skip the first step of the method and assume that we
 have a product  $P_1 = M\abs{}f_1$. In particular, we consider the
 automaton of Fig.~\ref{fig:rec-P}.

\begin{figure}[h]
    \centering
    \begin{center}
\begin{tikzpicture}[scale=\tiksize]
\tikzstyle{every node}+=[inner sep=0pt]
\draw [black] (15.3,-29.2) circle (3);
\draw (15.3,-29.2) node {$s_0$};
\draw [black] (27.5,-29.2) circle (3);
\draw (27.5,-29.2) node {$s_1$};
\draw [black] (27.5,-29.2) circle (2.4);
\draw [black] (6.6,-29.2) -- (12.3,-29.2);
\draw (6.1,-29.2) node [left] {$P\mbox{ }:=$};
\fill [black] (12.3,-29.2) -- (11.5,-28.7) -- (11.5,-29.7);
\draw [black] (13.977,-26.52) arc (234:-54:2.25);
\fill [black] (16.62,-26.52) -- (17.5,-26.17) -- (16.69,-25.58);
\draw [black] (17.538,-27.231) arc (120.15965:59.84035:7.687);
\fill [black] (25.26,-27.23) -- (24.82,-26.4) -- (24.32,-27.26);
\draw [black] (25.303,-31.213) arc (-58.88262:-121.11738:7.552);
\fill [black] (17.5,-31.21) -- (17.92,-32.05) -- (18.44,-31.2);
\end{tikzpicture}
\end{center}
    \caption{A toy product $P_1 = M\abs{}f_1$.
    $P_1$ has two states named $s_0$ and $s_1$. $s_0$ is the initial state and $s_1$ is the accepting state.
	The transition labels from $\Sigma$ are not needed and hence are omitted.
}
    \label{fig:rec-P}
\end{figure}
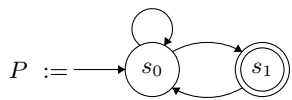

\begin{figure}[h]
    \centering
\begin{center}
\begin{tikzpicture}[scale=\tiksize]
\tikzstyle{every node}+=[inner sep=0pt]
\draw [black] (3,-32.8) circle (3);
\draw (3,-32.8) node {$s_0$};
\draw [black] (13.6,-32.8) circle (3);
\draw (13.6,-32.8) node {$s_1$};
\draw [black] (13.6,-32.8) circle (2.4);
\draw [black] (24.3,-32.8) circle (3);
\draw (24.3,-32.8) node {$s_0$};
\draw [black] (34.9,-32.8) circle (3);
\draw (34.9,-32.8) node {$s_0$};
\draw [black] (55.7,-32.8) circle (3);
\draw (55.7,-32.8) node {$s_0$};
\draw [black] (66.3,-32.8) circle (3);
\draw (66.3,-32.8) node {$s_1$};
\draw [black] (66.3,-32.8) circle (2.4);
\draw [black] (76.9,-32.8) circle (3);
\draw (76.9,-32.8) node {$s_0$};
\draw [black] (8.2,-25.3) circle (3);
\draw (8.2,-25.3) node {$s_0$};
\draw [black] (24.3,-25.3) circle (3);
\draw (24.3,-25.3) node {$s_1$};
\draw [black] (24.3,-25.3) circle (2.4);
\draw [black] (40,-25.3) circle (3);
\draw (40,-25.3) node {$s_0$};
\draw [black] (61,-25.3) circle (3);
\draw (61,-25.3) node {$s_0$};
\draw [black] (76.9,-25.3) circle (3);
\draw (76.9,-25.3) node {$s_1$};
\draw [black] (76.9,-25.3) circle (2.4);
\draw [black] (16.2,-17.7) circle (3);
\draw (16.2,-17.7) node {$s_0$};
\draw [black] (40,-17.7) circle (3);
\draw (40,-17.7) node {$s_1$};
\draw [black] (40,-17.7) circle (2.4);
\draw [black] (68.9,-17.7) circle (3);
\draw (68.9,-17.7) node {$s_0$};
\draw [black] (28.1,-9.9) circle (3);
\draw (28.1,-9.9) node {$s_0$};
\draw [black] (68.9,-9.9) circle (3);
\draw (68.9,-9.9) node {$s_1$};
\draw [black] (68.9,-9.9) circle (2.4);
\draw [black] (48.3,-3.1) circle (3);
\draw (48.3,-3.1) node {$s_0$};
\draw [black] (45.2,-32.8) circle (3);
\draw (45.2,-32.8) node {$s_1$};
\draw [black] (45.2,-32.8) circle (2.4);
\draw [black] (6.49,-27.77) -- (4.71,-30.33);
\fill [black] (4.71,-30.33) -- (5.58,-29.96) -- (4.75,-29.39);
\draw [black] (9.95,-27.73) -- (11.85,-30.37);
\fill [black] (11.85,-30.37) -- (11.79,-29.42) -- (10.97,-30.01);
\draw [black] (24.3,-28.3) -- (24.3,-29.8);
\fill [black] (24.3,-29.8) -- (24.8,-29) -- (23.8,-29);
\draw [black] (38.31,-27.78) -- (36.59,-30.32);
\fill [black] (36.59,-30.32) -- (37.45,-29.94) -- (36.62,-29.38);
\draw [black] (59.27,-27.75) -- (57.43,-30.35);
\fill [black] (57.43,-30.35) -- (58.3,-29.99) -- (57.48,-29.41);
\draw [black] (62.73,-27.75) -- (64.57,-30.35);
\fill [black] (64.57,-30.35) -- (64.52,-29.41) -- (63.7,-29.99);
\draw [black] (76.9,-28.3) -- (76.9,-29.8);
\fill [black] (76.9,-29.8) -- (77.4,-29) -- (76.4,-29);
\draw [black] (14.03,-19.77) -- (10.37,-23.23);
\fill [black] (10.37,-23.23) -- (11.3,-23.05) -- (10.61,-22.32);
\draw [black] (18.39,-19.75) -- (22.11,-23.25);
\fill [black] (22.11,-23.25) -- (21.87,-22.34) -- (21.19,-23.06);
\draw [black] (40,-20.7) -- (40,-22.3);
\fill [black] (40,-22.3) -- (40.5,-21.5) -- (39.5,-21.5);
\draw [black] (66.74,-19.78) -- (63.16,-23.22);
\fill [black] (63.16,-23.22) -- (64.09,-23.03) -- (63.39,-22.31);
\draw [black] (71.07,-19.77) -- (74.73,-23.23);
\fill [black] (74.73,-23.23) -- (74.49,-22.32) -- (73.8,-23.05);
\draw [black] (25.59,-11.54) -- (18.71,-16.06);
\fill [black] (18.71,-16.06) -- (19.65,-16.04) -- (19.1,-15.2);
\draw [black] (30.61,-11.54) -- (37.49,-16.06);
\fill [black] (37.49,-16.06) -- (37.1,-15.2) -- (36.55,-16.04);
\draw [black] (45.46,-4.06) -- (30.94,-8.94);
\fill [black] (30.94,-8.94) -- (31.86,-9.16) -- (31.54,-8.21);
\draw [black] (51.15,-4.04) -- (66.05,-8.96);
\fill [black] (66.05,-8.96) -- (65.45,-8.23) -- (65.13,-9.18);
\draw [black] (68.9,-12.9) -- (68.9,-14.7);
\fill [black] (68.9,-14.7) -- (69.4,-13.9) -- (68.4,-13.9);
\draw [black] (41.71,-27.77) -- (43.49,-30.33);
\fill [black] (43.49,-30.33) -- (43.45,-29.39) -- (42.62,-29.96);
\draw [black] (48.3,1.5) -- (48.3,0);
\fill [black] (48.3,0) -- (48.8,0.8) -- (47.8,0.8);
\draw (48.3,3.7) node {$\textit{Tree}(P)$};
\draw (-4,3.7) node {$\underline{n}$};
\draw (-4,-3.1) node {$0$};
\draw (-4,-9.9) node {$1$};
\draw (-4,-17.7) node {$2$};
\draw (-4,-25.3) node {$3$};
\draw (-4,-32.8) node {$4$};
\draw [dotted] (-1.5,5) -- (-1.5,-34);
\draw [dotted] (-1.1,-3.1) -- (45.3,-3.1);
\draw [dotted] (-1.1,-9.9) -- (25.1,-9.9);
\draw [dotted] (-1.1,-17.7) -- (13.2,-17.7);
\draw [dotted] (-1.1,-25.3) -- (5.2,-25.3);
\draw [dotted] (-1.1,-32.8) -- (0,-32.8);
\end{tikzpicture}
\end{center}
    \caption{Partial unfolding of the automaton of Fig.~\ref{fig:rec-P} into a tree up to depth 4. The column
    labeled $n$ denotes the number of transitions taken.}
    \label{fig:rec-P-tree}
\end{figure}
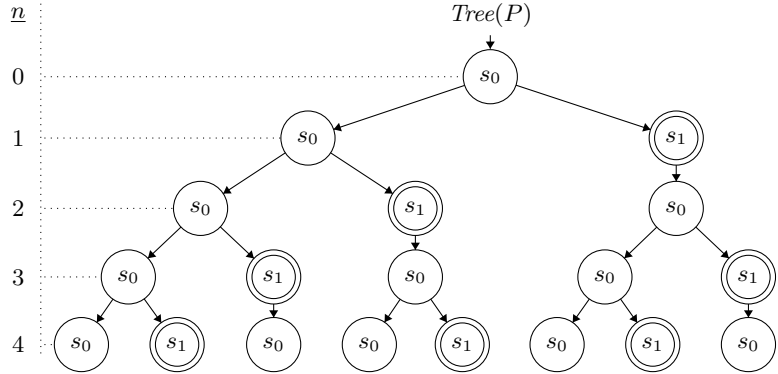

From the automaton of Fig.~\ref{fig:rec-P} we extract the following recurrence relations:
\begin{align}
\beta_{n+1}^{s_0} &= \beta_n^{s_0} + \beta_n^{s_1},
&& \beta_0^{s_0} = 1
\label{eqn:rec0}\\
\beta_{n+1}^{s_1} &= \beta_n^{s_0},
&& \beta_0^{s_1} = 0
\label{eqn:rec1}\\
\alpha_{n+1}^{s_0} &= \alpha_n^{s_0} + \alpha_n^{s_1},
&& \alpha^{s_0}_0 = 0
\label{eqn:rec2}\\
\alpha_{n+1}^{s_1} &= \alpha_n^{s_0} + \beta_n^{s_0},
&& \alpha^{s_1}_0 = 0
\label{eqn:rec3}\\
\alpha_{n}
&= \alpha_n^{s_0} + \alpha_n^{s_1},
&& \alpha_{\emptyset} = 0
\label{eqn:rec4}
\end{align}
where:
\begin{itemize}
\item $\beta_n^q$ is the total number of $n$-length paths through $P_1$ ending in state $q$,
    e.g., $\beta_0^{s_0} =  1$, $\beta_0^{s_1} =  0$, $\beta_3^{s_0} =  3$, $\beta_4^{s_1} =  3$.
	We encourage the reader to refer to Fig.~\ref{fig:rec-P-tree} and convince themselves that these examples hold.

\item $\alpha_n^q$ is the total number of accepting states visited along all
    $n$-length paths through $P_1$ restricted to paths terminating in state $q$,
    e.g., $\alpha_1^{s_0} = 0$, $\alpha_1^{s_1} = 1$, $\alpha_3^{s_0} = 2$.
\item $\alpha_n$
is the total number of accepting states visited along all
    $n$-length paths through $P_1$,
    e.g., $\alpha_0 = 0$, $\alpha_1 = 1$, $\alpha_2 = 2$, $\alpha_3 = 5$,
    $\alpha_4 = 10$.
\item $\alpha_\emptyset$ is a dummy variable representing the initial condition of $\alpha_n$. Notice that the $\alpha_n$ term of the recurrence is unique in that no other term depends on it.
\end{itemize}

\noindent
We determine each equation of the example recurrence relation as follows:

	Equations (\ref{eqn:rec0}) capture the number of paths of a certain
    length ending in state $s_0$. The initial value $\beta_0^{s_0}$ is $1$ because
    $s_0$ is an initial state. Otherwise, notice that $s_0$ has two
    predecessors: $s_0$ and $s_1$. To walk an $(n+1)$-length path ending in
    $s_0$, it is necessary and sufficient to walk an $n$-length path to one of
    its predessors and then take one more step. Hence, we compute
    $\beta_{n+1}^{s_0}$ as the sum of
    $\beta_{n}^{s_0}$ and $\beta_{n}^{s_1}$.
    Analogous reasoning yields Equations (\ref{eqn:rec1}); notice the initial
    value $\beta_0^{s_1}$ is 0 since $s_1$ is not an initial state.

	Equations (\ref{eqn:rec2}) capture the number of accepting states
    visited along all paths of a certain length ending in state $s_0$.
    Importantly, $s_0$ is not an accepting state. Therefore, adding it to an
    $n$-length path will not change the number of accepting states visited
    along that path. Hence, as with $\beta$, we can compute
    $\alpha_{n+1}^{s_0}$ as the sum of $\alpha_{n}^{s_0}$ and
    $\alpha_{n}^{s_1}$. The initial value $\alpha_0^{s_0}$ is $0$ because $s_0$ is
    an initial state, but not an accepting state.

	Equations (\ref{eqn:rec3}) capture the number of accepting states
    visited along all paths of a certain length ending in state $s_1$. Unlike
    $s_0$, the state $s_1$ is an accepting state. Therefore, the $(n+1)$th step
    contributes to the number of accepting states visited, in particular for
    each path it will increase the count by one. There are $\beta_{n}^{s_0}$
    such paths, hence the inclusion of that term in addition to the $\alpha$ of
    the predecessor $s_0$. The initial value $\alpha_0^{s_1}$ is $0$ because $s_1$
    is an accepting state, but not an initial state.

	Equations (\ref{eqn:rec4}) capture the accepting states along all
    paths of a certain length. The initial value $\alpha_\emptyset$ is
    irrelevant; we use 0 for simplicity. Otherwise, this equation merely
    captures the fact that we can partition the paths of length $n$ based on
    which state they end in and take a sum over that partition to compute a
    value over all paths.

We can represent these recurrence relation as a matrix-vector pair $(\xi_1,v_1)$, where:
\begin{equation*}
v_1 =
\begin{bmatrix}
\alpha_{\emptyset} \\ \alpha_0^{s_0} \\ \alpha_0^{s_1} \\ \beta_0^{s_0} \\ \beta_0^{s_1}
\end{bmatrix}
=
\begin{bmatrix}
0 \\ 0 \\ 1 \\ 0 \\ 0
\end{bmatrix}
\ \mbox{ and } \
\xi_1 =
\begin{bmatrix}
0 & 1 & 1 & 0 & 0 \\
0 & 1 & 1 & 0 & 0 \\
0 & 1 & 0 & 1 & 0 \\
0 & 0 & 0 & 1 & 1\\
0 & 0 & 0 & 1 & 0
\end{bmatrix}
\end{equation*}
E.g. row 1 of $\xi_1$ indicates which terms are required to compute $\alpha_n$.

\subsubsection{In general:}
The key
to  generalizing the above method is the set of predecessors for each state and how each
term should be computed using the predecessor terms. Not shown in this example
is the case where a state $q$ is both an initial state and an accepting state. In
that case $\alpha_0^q$ is $1$.
Also there is at most one transition between two
states in this example. In general, there may be multiple transitions between
two states (with different labels). In that case, the equations will include
factors in front of the $\alpha$ and $\beta$ terms. In particular,
$$\beta_{n+1}^{q'} = \sum_{q\in Q} t_{q, q'}\cdot\beta_{n}^q$$
where $t_{q, q'}$ is the number of transition labels that transition from $q$ to $q'$ (Note: $t_{q, q'}$ is 0 if $q$ is not a predessor of $q'$). Likewise:
$$\alpha_{n+1}^{q'} =
\sum_{q\in Q} (t_{q, q'}\cdot\alpha_{n}^q) + (t^*_{q, q'}\cdot\beta_{n}^q)$$
where $t^*_{q, q'}$ is $t_{q, q'}$ when $q'$ is an accepting state and 0 otherwise.

Now we explain the recurrence relation extraction algorithm in general. Let $P = M\abs{}f$ be the synchronous product of some finite LTS $M$ and some DFA $f$.
We explain how to extract both the recurrence matrix $\xi$ and the initial condition vector $v$ from $P$.

In what follows, we assume that $P$ has $N$ states indexed by the set $\{1,...,N\}$.
We first define a matrix that encodes the transition relation of $P$:
\begin{definition}
	We define the $N\times N$ {\em predecessor matrix}, denoted \D, by its entries. We denote the entry in the $i$th row and $j$th
    column as $\D_{ij}$. Define
    $\D_{ij}$ to be the number of transitions from state $j$ to state $i$ in $P$.
\end{definition}
Next, we define a matrix that encodes the accepting states of $P$:
\begin{definition}
    We define the $N\times N$ {\em accepting matrix}, denoted \A, so that $\A_{ij} = \D_{ij}$ if state $i$ of $P$ is an accepting state. Otherwise, $\A_{ij} = 0$.
\end{definition}
We are now able to define the recurrence matrix $\xi$:
\begin{definition}
    The {\em recurrence matrix} of $P$ is the $(2N+1)\times(2N+1)$ matrix
    $$
    \xi =
    \begin{bmatrix}
        0       & \hat{1} & \hat{0} \\
        \hat{0} & \D      & \A\\
        \hat{0} & \mathbf{0} & \D
    \end{bmatrix}
    $$
    where $\hat{0}$ and $\hat{1}$ are $n$-dimensional vectors of 0's and 1's respectively and where $\mathbf{0}$ is an $n\times n$ matrix of 0's.
\end{definition}

We now explain how to extract the initial condition vector $v$ from $P$. We first introduce some notation.
    For convenience, we vectorize the $\alpha_n^q$ and $\beta_n^q$ terms.
	Let
    $\valpha n := (\alpha_n^1,...,\alpha_n^N)^T$
    and
    $\vbeta n := (\beta_n^1,...,\beta_n^N)^T$.
Then, the two vectors $\valpha 0$ and $\vbeta 0$ capture the initial conditions of terms
$\alpha_n^i$ and $\beta_n^i$ in the recurrence relation, and we can construct the $2N+1$ dimensional vector $v$ by combining
$\valpha 0$ and $\vbeta 0$ along with $\alpha_\emptyset = 0$, namely, $v := (\alpha_\emptyset,\valpha 0, \vbeta 0)^T$.

The vectors $\valpha 0$ and $\vbeta 0$
 are extracted from $P$ as follows:

    (1) The $i$th entry of $\valpha 0$ is 1 if and only if  state   $i$ of $P$ is both an accepting state and an initial state. Otherwise, that entry of $\valpha 0$ is 0.
    (2) The $i$th entry of $\vbeta 0$ is 1 if and only if state $i$ of $P$ is an initial state. Otherwise, that entry of $\vbeta 0$ is 0.

The following two statements
(proven in Appendix~\ref{sec_pf-correct})
capture the correctness of our reduction.

\begin{theorem}
    \label{thm:main-correctness}
    Let $\alpha$ and $\beta$ be the recurrence relation terms for the product $M\abs{}f_i$, as constructed above.
    Then for all $n\geq 0$,
    $\xi_i^{n+1}v_i =
    \begin{bmatrix}
        \alpha_n \\ \valpha {n+1} \\ \vbeta {n+1}
    \end{bmatrix}
    $. And hence $(\xi_i^{n+1}v_i)_0 = \alpha_n = \Xsum{\img{f}{M_n}}{i}$.
\end{theorem}

\begin{corollary}
    \label{cor:main-correctness}
    Let $\xi_i$ and $v_i$ be the recurrence matrices and initial condition vectors for the products $M\abs{}f_i$, for $i=1,...,d$, as constructed above.
	Then
    $$\agg_f(M) = \lim_{n\to\infty}
    h((\xi_1^{n+1}v_1)_0, (\xi_2^{n+1}v_2)_0, ..., (\xi_d^{n+1}v_d)_0)$$
\end{corollary}

\subsection{Step~\ref{step_matrix}: Matrix Analysis}
\label{sec_matrix_analysis}

Next we will discuss two methods for solving the matrix analysis problem. One of these methods is {\em symbolic} and the other {\em numerical}.
We illustrate them by continuing with the example of Fig.~\ref{fig:rec-P}.
We have constructed $g_1(n) = (\xi_1^{n+1}v_1)_0$. For sake of example, let us assume that $\xi_1 = \xi_2$ and that $v_2 = \xi_1v_1$, so $g_2(n) = g_1(n+1)$.
Let us also assume that $h(g_1(n), g_2(n)) = g_1(n) / g_2(n)$.

\subsubsection{Symbolic Method:}
The first step of the symbolic method is to compute closed-form expressions for each $g_i$. Tools such as Mathematica can do this using Jordan decomposition~\cite{hefferon2020linear}.
We omit the details.
The result is:
{\small
\begin{align*}
    g_1(n) =
        \frac{1}{25\cdot 2^{(1 + n)}}
        \bigg(
        4 \sqrt{5} \negc{1}^n - 4 \sqrt{5} \posc{1}^n
        - 5 \negc{1}^n n + 5 \sqrt{5} \negc{1}^n n
        - 5 \posc{1}^n n - 5 \sqrt{5} \posc{1}^n n\bigg)
\end{align*}}
where $\posc{1} := 1 + \sqrt{5}$ and $\negc{1} := 1 - \sqrt{5}$. As mentioned, $g_2(n) = g_1(n+1)$.

Once we have the closed-form expressions, we can ask Mathematica to solve the limit; it does so easily:
$\lim_{n\to\infty}g_1(n)/g_2(n) = 2/(1 + \sqrt{5})$.
This value may be readily familiar to some as the reciprocal of the golden ratio.
Tools such as Mathematica can solve a broad class of limits using, e.g., Gruntz's method~\cite{gruntz_1996}.

Computing the Jordan decomposition is currently the bottleneck for the symbolic method. Our experiments with Mathematica suggest that it cannot compute the Jordan decomposition for even moderately sized matrices,
the runtime being exponential in the dimension of the matrix.
There have been several recent attempts to improve the state of the art in Jordan decomposition~\cite{10032513,shi2023computation} and we are hopeful that this subproblem will soon be feasible to compute for large matrices.

\subsubsection{Numerical Method:}
In this method, we compute $h(g_1(K), g_2(K))$ for large $K$, which we call a {\em $K$-approximation}.
Although we have not yet established an error bound on the difference between the $K$-approximation and the true value of the limit, the $K$-approximation appears to converge relatively quickly. For instance, in the case of Example~\ref{fig:rec-P}, the $K$-approximation for $K=15$ and $K=20$ are 0.6180344 and 0.6180339 respectively, which do not differ until the seventh decimal place.
Our current approach is to compute the $K$-approximation for, e.g., $K=8192$ and $K=9000$ and determine at which decimal place they differ to establish the precision of the $K$-approximation for $K=9000$. We can also plot intermediate $K$-approximations against $K$.

A naive implementation of $K$-approximation does not scale. Instead, we use the standard
{\em exponentiation by squaring} technique to quickly compute $K$-approximations for large $K$. For example, to compute $M^{11}$ for some matrix $M$, it suffices to compute $M^2, M^4,$ and $M^8$, since
$M^{11} = M\cdot M^2\cdot M^8$. Note that $M^4 = (M^2)^2$ and $M^8 = (M^4)^2$, hence the name {\em exponentiation by squaring}. We need only compute $\log K$ squares and combine them per the binary representation of $K$.
Furthermore, in our implementation, we found that we needed large datatypes (128 bit) to represent the entries of the matrix.
As matrix power for large datatypes appears to not be implemented in the linear algebra library we used (numpy),
we implemented this operation ourselves.

\subsubsection{Comparison:}
The symbolic method gives an exact, symbolic representation of the fitness score, but unfortunately does not yet scale well, as we shall see from the experiments in Section~\ref{sec_casestudy} that follows. The numerical approach on the other hand can compute in seconds an approximation of the fitness score. As we shall show, these approximations are precise enough to distinguish between systems of different fitness.

\section{Case Studies}
\label{sec_casestudy}

We evaluate our framework on three case studies, described in detail in the subsections that follow, and summarized in Table~\ref{case-study-results}.
The symbolic method did not terminate after an hour for the larger two case studies (2PC and ABP) due to limitations imposed by the state of the art in Jordan decomposition
(c.f. Section~\ref{sec_matrix_analysis}).
Therefore, Table~\ref{case-study-results} reports the results obtained by the numerical method.

In each case study we compute the fitness score for different system variants (column $M$).
Column $|M|$ represents the size (total number of states) of the system being measured, which is the product of all distributed processes.
Time refers to the total execution time, in seconds.
Column $\agg_f(M_{8192})$ refers to the $K$-approximation of the fitness score with $K=8192$, and likewise for $K=9000$.
As can be seen, the two approximations are very close within each row (identical up to at least the 3rd decimal point), which indicates convergence.
The reason we report the fitness score for K = 8192 instead of another number, say K = 8000 or K = 8500, is efficiency: 8192 the largest power of two less than 9000, and in order to compute the fitness score for K = 9000 we need to compute it anyway for K = 8192.
Our results can be reproduced using a publicly available artifact, which is structured, documented, and licensed for ease of repurposing \cite{artifactDOI}.

Let us remark that in the 2PC and ABP case studies, the systems being measured were automatically generated by a distributed protocol synthesis tool,
which is an improved version of the tool described in~\cite{ScenariosHVC2014,AlurTripakisSIGACT17}.
As our goal in this paper is fitness evaluation, we omit discussing the synthesis tool.
But, as mentioned in the introduction, evaluation of automatically synthesized systems is a promising application of our framework.

All case studies use the $\avgrate$ aggregate function.
Additionally, we use three variations of the fitness function in Fig.~\ref{fig:protocol-fitness}. This parametric fitness function suggests the possibility of constructing a library of general, reusable fitness functions. Although it was straightforward to construct fitness functions for our purposes, this library would further reduce that burden for users.

In the rest of this section we provide further details on each case study.
Some supporting figures and intermediate results are provided in
Appendix~\ref{sec_cstudy-overflow}.


\begin{table}[]
    \centering
\begin{tabular}{|l|l|r|r|r|r|}
\hline
case study &
  $M$ &
  $\abs{M}$ &
  total time (sec.) &
  $\agg_f(M_{8192})$ &
  $\agg_f(M_{9000})$ \\
\hline
simple comm. & good & 3   & 0.0052 & 0.249970  & 0.249972 \\
simple comm. & bad  & 5   & 0.006  & 0.138165 & 0.138168 \\
2PC          & H    & 58  & 0.41   & 0.0833   & 0.0832   \\
2PC          & A1   & 30  & 0.25   & 0.07856  & 0.07857  \\
2PC          & A2   & 25  & 0.1    & 0.0833   & 0.0832   \\
ABP          & HH   & 144 & 9.1    & 0.016864 & 0.016859 \\
ABP          & HA   & 144 & 8.6    & 0.015435 & 0.015430  \\
ABP          & AH   & 144 & 8.7    & 0.015218 & 0.015212 \\
ABP          & AA   & 144 & 8.6    & 0.01391  & 0.01390 \\
\hline
\end{tabular}
\caption{A summary of the numerical method results of the three case studies.
}
\label{case-study-results}
\end{table}

\subsection{Case Study \#1: Simple Communication Protocol}
This section treats the communication protocol presented in Example~\ref{ex_protocol}.
We instantiate the framework to measure the average rate at which send-ack sequences are executed and apply this instance of the framework to $M$ and $M'$ (Fig.~\ref{fig:protocol}).
The python representations of all simple communication protocol processes and fitness functions are available in \texttt{toy\_automata.py} of the artifact~\cite{artifactDOI}.

Recall that $\Sigma=\{s,t,a\}$.
Let $f_1(w) :=$ `the number of send-ack sequences of the form $st^*a$ in $w$'.
For instance (brackets $[$ and $]$ added for emphasis),
$f_1(aat [sa] [sta]a s[stta] stt[sa]) = 4$.
Additionally, let $f_2(w) := \abs{w}$ (the length of $w$) and let the fitness function be $f := \angles{f_1,f_2}$. The functions $f_1,f_2$ can be represented as the DFA shown in Fig.~\ref{fig:protocol-fitness}, with $L = \{s\}$ and $R = \{a\}$.
This fitness function is measuring the number of send-ack sequences per unit of discrete time, which is analogous to the traditional measure of throughput in distributed systems.

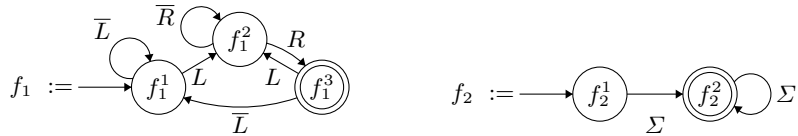
\begin{figure}
    \begin{subfigure}[b]{0.475\textwidth}
        \centering
\begin{center}
\begin{tikzpicture}[scale=\tiksize]
\tikzstyle{every node}+=[inner sep=0pt]
\draw [black] (31,-28.1) circle (3);
\draw (31,-28.1) node {$f_1^1$};
\draw [black] (40,-22.8) circle (3);
\draw (40,-22.8) node {$f_1^2$};
\draw [black] (49.1,-28.1) circle (3);
\draw (49.1,-28.1) node {$f_1^3$};
\draw [black] (49.1,-28.1) circle (2.4);
\draw [black] (33.59,-26.58) -- (37.41,-24.32);
\fill [black] (37.41,-24.32) -- (36.47,-24.3) -- (36.98,-25.16);
\draw (35.5,-25.95) node [below] {$L$};
\draw [black] (37.044,-23.238) arc (306.1666:18.1666:2.25);
\draw (32.81,-20.06) node [left] {$\overline{R}$};
\fill [black] (37.85,-20.72) -- (37.79,-19.78) -- (36.98,-20.37);
\draw [black] (42.961,-23.21) arc (73.78871:45.77674:10.33);
\fill [black] (47.28,-25.73) -- (47.06,-24.81) -- (46.36,-25.53);
\draw (46.22,-23.7) node [above] {$R$};
\draw [black] (28.176,-27.123) arc (278.64042:-9.35958:2.25);
\draw (24.76,-22.82) node [above] {$\overline{L}$};
\fill [black] (30.06,-25.26) -- (30.43,-24.4) -- (29.44,-24.55);
\draw [black] (22.1,-28.1) -- (28,-28.1);
\draw (21.6,-28.1) node [left] {$f_1\mbox{ }:=$};
\fill [black] (28,-28.1) -- (27.2,-27.6) -- (27.2,-28.6);
\draw [black] (46.51,-26.59) -- (42.59,-24.31);
\fill [black] (42.59,-24.31) -- (43.03,-25.14) -- (43.54,-24.28);
\draw (43.66,-25.95) node [below] {$L$};
\draw [black] (46.342,-29.273) arc (-71.32758:-108.67242:19.653);
\fill [black] (33.76,-29.27) -- (34.36,-30) -- (34.68,-29.06);
\draw (40.05,-30.81) node [below] {$\overline{L}$};
\end{tikzpicture}
\end{center}
    \end{subfigure}
    \begin{subfigure}[b]{0.475\textwidth}
        \centering
\begin{center}
\begin{tikzpicture}[scale=\tiksize]
\tikzstyle{every node}+=[inner sep=0pt]
\draw [black] (31,-28.1) circle (3);
\draw (31,-28.1) node {$f_2^1$};
\draw [black] (43.2,-28.1) circle (3);
\draw (43.2,-28.1) node {$f_2^2$};
\draw [black] (43.2,-28.1) circle (2.4);
\draw [black] (22.1,-28.1) -- (28,-28.1);
\draw (21.6,-28.1) node [left] {$f_2\mbox{ }:=$};
\fill [black] (28,-28.1) -- (27.2,-27.6) -- (27.2,-28.6);
\draw [black] (34,-28.1) -- (40.2,-28.1);
\fill [black] (40.2,-28.1) -- (39.4,-27.6) -- (39.4,-28.6);
\draw (37.1,-30.6) node [below] {$\Sigma$};
\draw [black] (45.88,-26.777) arc (144:-144:2.25);
\draw (50.45,-28.1) node [right] {$\Sigma$};
\fill [black] (45.88,-29.42) -- (46.23,-30.3) -- (46.82,-29.49);
\end{tikzpicture}
\end{center}
    \end{subfigure}
    \caption{The DFA representations of $f_1$ and $f_2$ for the case studies, parameterized by the set of labels $\Sigma$, as well as a set of {\em left endpoints} $L\subseteq\Sigma$ and {\em right endpoints} $R\subseteq\Sigma$. $\overline{L} = \Sigma\setminus L$ and likewise for $\overline{R}$.}
    \label{fig:protocol-fitness}
\end{figure}

As reported in Table~\ref{case-study-results}, the system that uses the good receiver has a fitness score of about $0.25$ and the system using the bad receiver a score of about $0.138$. These scores are {\em interpretable} in that they have units: send-ack sequences per unit of discrete time.
Hence, the framework deems the good receiver as more fit and this determination aligns with our intuitions.
Because this example is relatively small, Mathematica was able to compute the exact fitness scores of these systems. The system that uses the good receiver has a fitness score of exactly $1/4$ (obtained after 34 seconds) and the system that uses the bad receiver has a score of exactly $\frac{5-\sqrt{5}}{20}\approx 0.138$ (obtained after 563 seconds).

\subsection{Case Study \#2: Two Phase Commit (2PC)}
Two phase commit (2PC) is a protocol for making transactional changes to a distributed database atomically; if one sub-operation of the transaction is aborted at one remote database, so too must the sub-operations at all other remote databases.
Although each iteration of 2PC is terminating, it is typical to assume there will be infinitely many such iterations, and our model reflects this.
In our model of 2PC, a user initiates a transaction by synchronizing with a {\em transaction manager} on the label $x$. The transaction is complete when the transaction manager synchronizes with the user on label \textit{fail} or \textit{succ}. We omit the details of the intermediate exchanges between the transaction manager and database managers.
The python representations of all 2PC processes and fitness functions are available in \texttt{\_2pc\_automata.py} of the artifact~\cite{artifactDOI}.

The fitness function for this case study is as depicted in Fig.~\ref{fig:protocol-fitness}, with $L = \{x\}$, $R = \{\textit{fail},\textit{succ}\}$, and $\Sigma$ has a total of 18 labels. This fitness function measures the rate at which transactions are initiated and then completed.

We study three 2PC implementations, each using a different transaction manager LTS.
The system labeled H in Table~\ref{case-study-results} uses a previously manually constructed transaction manager that the synthesis tool was also able to discover automatically, while the systems labeled A1 and A2 use new transaction managers generated by the synthesis tool.
The automatically generated transaction managers  have 12 states each and it is therefore hard to tell at a glance which will give rise to the most efficient protocol.
Our tool automatically reports, in fractions of a second, a fitness score of about $0.083$ for both systems
 H and A2, and a score of about $0.079$ for system A1.
These fitness scores have units: transactions per unit time.
Hence, in the same amount of time, A1 completes about 5\% fewer transactions than H or A2.

\subsection{Case Study \#3: Alternating Bit Protocol (ABP)}
The Alternating Bit Protocol (ABP) allows reliable communication over an unreliable network. As with the prior two case studies, we use the fitness function depicted in Fig.~\ref{fig:protocol-fitness}, except with $L = \{\textit{send}\}, R = \{\textit{done}\}$, and $\Sigma$ of size 12. Similar to case study \#1 we are measuring the rate of send-done sequences.
The python representations of all ABP processes and fitness functions are available in \texttt{abp\_automata.py} of the artifact~\cite{artifactDOI}.

In~\cite{AlurTripakisSIGACT17}, the authors present a method to automatically synthesize (distributed) ABP sender and receiver processes. Here, we evaluate the fitness of the ABP variants that use these various synthesized processes.
Together the synthesized sender and receiver processes have 14 states, which again makes manual determinations about the fitness very challenging---even more so due to the distributed nature of the problem. It is no longer necessarily a question of which sender or receiver is better than the other sender or receiver, but a question of which combination of sender and receiver is best.
Once again, our framework allows to automatically make this determination in a matter of seconds.

The systems are ranked by fitness in the following order: HH, HA, AH, AA.
H stands for human-designed (and then also rediscovered during synthesis) and A stands for newly discovered during synthesis.
In this case study, the newly discovered processes do worse than the manually constructed processes.
The difference in fitness scores is meaningful: in the same amount of time, AA will complete about 18\% fewer sequences on average. AH and HA will both complete about 8.5\% fewer sequences than HH.

\section{Related Work}
\label{sec_related}

Our work is broadly related to the field of performance analysis and evaluation.
Mathematical models typically used there include Markov Chains, Markov Decision Processes, Markov Automata, queueing models, Petri nets, timed
or hybrid automata, etc., e.g., see~\cite{Baier2008PMC,1231550,1556551,CassandrasLafortune2021,DBLP:conf/date/FakihGFR13,10.1007/978-3-642-24310-3_1,DBLP:journals/fmsd/KwiatkowskaNPS06,7774639}.
Our approach differs as our mathematical framework uses neither timed nor probabilistic models such as the ones above.
Because we do not use stochastic models, our work is also different from the work on probabilistic verification, e.g., see~\cite{BaierCACM2010,Baier2008PMC,handbookMC_Baier2018,cauchi2018efficient,10.1007/978-3-319-10696-0_31}.
Our work also differs from performance analysis approaches that use max-plus algebra based frameworks such as the
real-time calculus, e.g., see~\cite{6728887,LU2012298,Thiele00RTC,DBLP:reference/crc/WandelerT05}.

Our work is also related to non-boolean interpretations of temporal semantics, such as the 5-valued robust temporal logic rLTL~\cite{DBLP:journals/tocl/AnevlavisPNT22,DBLP:conf/csl/TabuadaN16}. However, our motivation is performance comparisons rather than robustness.
Our framework also differs from that of signal temporal logic (STL)~\cite{8353490,10.1007/978-3-030-99524-9_15,a15040126,DBLP:journals/sttt/NickovicLMFU20,8603214,puranic2021learning,SALAMATI2021109781},
which is valued over real-time traces. Our framework is over discrete traces, although there have been recent STL extensions which handle both real and discrete time~\cite{10.1007/978-3-030-29662-9_4}.
In addition, our framework is parameterized by generic quantitative concepts (the fitness and aggregate functions and the comparison relation) that are
present neither in rLTL nor in STL or its variants.

Our work is closely related to the field of quantitative verification, synthesis, and games, e.g.,
see~\cite{DBLP:journals/tcs/AlfaroFHMS05,almagor2018equilibria,10.1007/978-3-319-11936-6_6,brihaye_et_al:LIPIcs:2015:5622,DBLP:conf/cav/CernyCHRS11,ChatterjeeDH10,DBLP:conf/qest/ChatterjeeAFHMS06,DBLP:journals/ife/Henzinger13}.
Typically, these works assign values to {\em weighted automata}.
These automata blend in a single model both the description of the system and the description of any performance or fitness functions associated with the system. In comparison, our framework decouples the description of the system (e.g., a plain LTS without any weights) from the description of the fitness function (e.g., a DFA). Our semantical framework is also very general and can handle multi-dimensional fitness functions and arbitrary aggregate functions, not just $\sup$, which is the only aggregate supported by these works.

Sensing cost, described in~\cite{almagor_et_al:LIPIcs:2014:4840}, measures how many signals each state of a system needs to observe in order to make a decision. The sensing cost of a run is the average sensing cost of the states visited along that run. The sensing cost of a system is the expected sensing cost along all runs. Finally, the sensing cost of a language is the minimal sensing cost across all automata that accept that language. It seems the primary focus of this work is to establish a complexity measure for languages, but it can certainly be used to compare two systems.
Sensing cost can be viewed as a particular fitness criteria, but it is a syntactic metric, whereas our framework considers semantic metrics. Sensing cost is syntactic in the sense that it is computed over runs of states rather than runs of transition symbols and it primarily uses quantities that are captured statically from the transition function.
Finally, sensing cost is measured solely with respect to input symbols and thus in some sense only measures how well a system can cope with the environment. On the other hand, our framework makes measurements over all symbols and can therefore yield results about e.g. the rate at which the system does a good thing.

Propositional quality, presented in \cite{DBLP:conf/icalp/AlmagorBK13}, is another way to measure the fitness of a system. Like our work, the framework used here is parameterized by arbitrary functions. Unlike our work that uses DFA's to specify fitness criteria, the authors formalize what they call quality using a quantitative variant of LTL. The emphasis of their paper is that this variant of LTL has computational problems that are analogous to those of traditional LTL and that these problems can be solved by natural extensions of non-quantitative algorithms without much if any additional run-time complexity overhead. There is no obvious reduction between our framework and propositional quality because the arbitrary functions introduced by the latter can only consider sub-traces of a fixed size and they do not take any limits, sup, inf, etc over this size parameter. In particular, it isn't obvious how propositional quality could express average throughput of a trace in the limit as we do for our case studies. Conversely, there is no obvious reduction of their work to our treatment of LTS with DFA fitness criteria, namely because their logic formulas induce a sort of recursive computation that can never be captured by a DFA. Their focus is on worst-case behavior whereas our focus has been on average-case behavior.

\section{Conclusions and Future Work}
\label{sec_concl}

We proposed a formal framework that assigns {\em fitness scores} to systems modeled as finite LTSs. The main novelty of our framework is that it {\em decouples} the description of the system from the set of domain-specific parameters such as fitness and aggregate functions, which determine the final fitness score. Furthermore, the user defines these fitness scores and aggregate functions over partial runs, which are easier for the user to reason about---our framework does the heavy lifting of extending this reasoning to infinite traces. This decoupling and finite reasoning make our framework more useable and its results more {\em interpretable}. Indeed, in all of our case studies the scores are not merely numbers; they have meaningful units, e.g., send-ack sequences per unit of time.

We used our framework to evaluate the automatically synthesized ABP protocols presented in~\cite{AlurTripakisSIGACT17}
as well as our own automatically synthesized 2PC protocols.
We showed that some of these protocols are better than others.
Inspired by this application, we plan to investigate the use of our framework in protocol synthesis, specifically in synthesizing protocols that not only satisfy a given correctness specification but are also {\em optimal} with respect to a fitness score.

We are also actively exploring ways to improve the scalability of the symbolic method. In particular, we may be able to feasibly compute a simplified version of the the recurrence matrix $\xi_i$ without sacrificing the accuracy of the final computed limit.
Additionally, we would like to generalize our method to aggregates like $\min/\max$, which do not have conditional representations, and to systems that cannot be represented as finite labeled transition systems.
We suspect that best/worst-case analysis reduces to the minimal cost-to-time ratio problem~\cite{Lawler1972}, but in general aggregates with no conditional representation may be more challenging.

\subsubsection{Acknowledgements}
Derek Egolf's research has been initially supported by a Northeastern University PhD fellowship.
This material is based upon work supported by the National Science Foundation Graduate Research Fellowship under Grant No. (1938052).
Any opinion, findings, and conclusions or recommendations expressed in this material are those of the authors(s) and do not necessarily reflect the views of the National Science Foundation.

%
%
%
\bibliographystyle{splncs04}
\bibliography{biblio}
%

\appendix
\captionsetup{belowskip=-10pt}
\section{Appendix}
\subsection{Example Illustrating Oscillating Fitness Score}
\begin{example}[Oscillating Fitness Score]
    \label{ex:osc-fitness}
    Here we provide an example of a system which does not have a well-defined fitness score \wrt rate of \$. Let $M = \{x\}$ be a system with just one trace. Below are some example of the prefixes of this $x$. This trace is pathological, unlikely to be seen by itself in the real world. We craft it specifically so that for all prefixes of $x$, there are two longer prefixes, $x_1$ and $x_2$, such that $\textit{rate}_\$(x_1) = 1/2$ and $\textit{rate}_\$(x_2) = 3/4$. It follows immediately that this fitness score oscillates between 1/2 and 3/4 in the limit.

    \begin{align*}
        &\text{Prefix} & \#_\$&:\#_0 && \textit{rate$_\$$}
        \\
        &0\$ & 1&:1 && 1/2
        \\
        &0\$\$\$ & 3&:1 && 3/4
        \\
        &0\$\$\$00 & 3&:3 && 1/2
        \\
        &0\$\$\$00\$\$\$\$\$\$ & 9&:3 && 3/4
        \\
        &0\$\$\$00\$\$\$\$\$\$000000 & 9&:9 && 1/2
        \\
        &0\$\$\$00\$\$\$\$\$\$000000\$\$\$\$\$\$\$\$\$\$\$\$\$\$\$\$\$\$ & 27&:9 && 3/4
        \\
        &\vdots
    \end{align*}
\end{example}

\subsection{Fitness-Optimal Synthesis Problem}
\label{sec_fitnessoptimalsynthesis}

In addition to the fitness score evaluation and comparison problems considered at the end of Section~\ref{sec_framework},
we can also consider the following {\it fitness-optimal synthesis} problem.
First, we define the following notion of fitness-optimality:

\begin{definition}[Fitness-optimality]
    Let $F = \angles{\Sigma, f,\agg,\cleq,\phi}$. We say that system $M^*$ is optimal \wrt $F$ if $\sat{M^*}{\phi}$ and
	for all systems $M$ such that $\sat{M}{\phi}$, we have $\agg_f M^* \cleq \agg_f M$.
    \label{def:optimality}
\end{definition}

We can then consider the following synthesis problem:
{\it Given fitness function $f$, aggregate function $\agg$, comparison relation $\cleq$, and specification $\phi$, compute,
if it exists, a system which is fitness-optimal with respect to $F = \angles{\Sigma, f,\agg,\cleq,\phi}$.} It is possible that such a system either does not exist or is not unique.

Studying fitness-optimal synthesis is beyond the scope of the current paper and is left for future work.

\subsection{Example of Definition~\ref{def:income-rate} Continued}
\label{sec_example_exrateDFA_continued}

    Consider again the fitness function represented by DFA $f_1$ and $f_2$ shown in
    Fig.~\ref{fig:fit-repr}.

    On input $\$0\$\$0$, $f_1$ visits the sequence of states (where circles denote accepting states)
    $$f_1^0,\circled{$f_1^1$},f_1^0,\circled{$f_1^1$},\circled{$f_1^1$},f_1^0$$
    and therefore, as desired, $f_1(\$0\$\$0) = 3$. On the same input $f_2$ visits
    $$f_2^0,\circled{$f_2^1$},\circled{$f_2^1$},\circled{$f_2^1$},\circled{$f_2^
    1$},\circled{$f_2^1$}$$
    and therefore $f_2(\$0\$\$0) = 5$. Hence, $f(\$0\$\$0) = (3,5)$, which is an
    analog for the rate $3/5$.

\subsection{Proof of Correctness}
\label{sec_pf-correct}
We  prove that the recurrence matrix and the initial condition vector enable us to compute $\Xsum{\img{f}{M_n}}{i}$. This section aims to show $\Xsum{\img{f}{M_n}}{i} = (\xi_i^{n+1} v_i)_0$. We will first prove several lemmas. We fix $i$ and write $\xi$ and $v$ rather than $\xi_i$ and $v_i$

\begin{lemma}
    For all $n \geq 0$, $\vbeta{n+1} = \D\vbeta n$
    \label{lemma:beta-D}
\end{lemma}
\begin{proof}
    Consider the following derivation.
    \begin{align*}
        \D\vbeta n
        =
        \begin{bmatrix}
        \D_1\vbeta n \\ \vdots \\ \D_N\vbeta n
        \end{bmatrix}
        =
        \begin{bmatrix}
        \displaystyle\sum_{i=1}^N \D_{1,i}\beta_n^i
        \\ \vdots \\
        \displaystyle\sum_{i=1}^N \D_{N,i}\beta_n^i
        \end{bmatrix}
        =
        \begin{bmatrix}
            \displaystyle\sum_{i=1}^N t_{1,i} \beta_n^i
            \\ \vdots \\
            \displaystyle\sum_{i=1}^N t_{N,i} \beta_n^i
        \end{bmatrix}
        \overeq{(*)}
        \begin{bmatrix}
            \beta_{n+1}^1
            \\ \vdots \\
            \beta_{n+1}^N
        \end{bmatrix}
        = \vbeta{n+1}
    \end{align*}
    where $t_{k,i}$ is the number of transitions from $i$ to $j$.

    All of these steps follow from definition, but step $(*)$ is worth clarifying. If $i$ is a predessor of $j$, then all $n$-length paths leading to $i$ can be extended by next going to state $j$. Hence, the number of $(n+1)$-length paths ending in $j$ is equal to
    $$\displaystyle\sum_{i=1}^N t_{j,i} \beta_n^i$$
\end{proof}

\begin{lemma}
    For all $n \geq 0$, $\valpha{n+1} = \D\valpha n + \A\vbeta n$
    \label{lemma:alpha-AD}
\end{lemma}
\begin{proof}
    This lemma follows from a very similar derivation to that in Lemma \ref{lemma:beta-D}. We omit it.
\end{proof}

Armed with these lemmas, we can prove the primary result of this paper, first stated in Section~\ref{sec_recur_rel}.

\subsubsection{Proof of Theorem~\ref{thm:main-correctness}.}
\begin{proof}[By induction]
    We first prove the base case where $n = 0$.
    \begin{align*}
        \xi v
        \overeq{(\text{Def})}
        \begin{bmatrix}
            0       & \hat{1} & \hat{0} \\
            \hat{0} & \D      & \A\\
            \hat{0} & \mathbf{0} & \D
        \end{bmatrix}
        \begin{bmatrix}
            0 \\ \valpha 0 \\ \vbeta 0
        \end{bmatrix}
        =
        \begin{bmatrix}
            \sum\valpha 0 \\ \D\valpha 0 + \A\vbeta 0 \\ \D\vbeta 0
        \end{bmatrix}
        \overeq{(\text{L\ref{lemma:beta-D}, L\ref{lemma:alpha-AD}})}
        \begin{bmatrix}
            \alpha_0 \\ \valpha 1 \\ \vbeta 1
        \end{bmatrix}
    \end{align*}
    and now we prove the inductive case.
    \begin{align*}
        \xi^{n+1}v
        =
        \xi\xi^{n}v
        \overeq{(\text{IH})}
        \begin{bmatrix}
            0       & \hat{1} & \hat{0} \\
            \hat{0} & \D      & \A\\
            \hat{0} & \mathbf{0} & \D
        \end{bmatrix}
        \begin{bmatrix}
            \alpha_{n-1} \\ \valpha n \\ \vbeta n
        \end{bmatrix}
        =
        \begin{bmatrix}
            \sum\valpha n \\ \D\valpha n + \A\vbeta n \\ \D\vbeta n
        \end{bmatrix}
        \overeq{(\text{L\ref{lemma:beta-D}, L\ref{lemma:alpha-AD}})}
        \begin{bmatrix}
            \alpha_n \\ \valpha {n+1} \\ \vbeta {n+1}
        \end{bmatrix}
    \end{align*}
\end{proof}

\subsubsection{Proof of Corollary~\ref{cor:main-correctness}.}
\begin{proof}
For each $f_i\abs{}M$, $(\xi_i^{n+1}v_i)_0 = \alpha_n$ by Theorem~\ref{thm:main-correctness}. By definition, $\alpha_n$ is the total number of accepting states visited across all length $n$ paths through $M\abs{}f_i$. Therefore, $(\xi_i^{n+1}v_i)_0 = \Xsum{\img{f}{M_n}}{i}$. The corollary follows:
\begin{align*}
    \agg_f(M)
    &= \lim_{n\to\infty}
    h(\Xsum{\img{f}{M_n}}{1}, \Xsum{\img{f}{M_n}}{2}, ... ,\Xsum{\img{f}{M_n}}{d})\\
    &= \lim_{n\to\infty}
    h((\xi_1^{n+1}v_1)_0, (\xi_2^{n+1}v_2)_0, ..., (\xi_d^{n+1}v_d)_0)
\end{align*}
\end{proof}

\subsection{Case Studies: Intermediate Results and Additional Figures}
\label{sec_cstudy-overflow}

We first point to the intermediate results for the system using the good receiver for the toy communication protocol of Example~\ref{fig:protocol}.
Figs.~\ref{fig:fitness-product-product} and ~\ref{fig:rec-matrix} are the intermediate results of steps and 1 and 2 of the algorithm, respectively.
Fig.~\ref{fig:Xhat1} shows the symbolic expressions used to compute the limit exactly. Likewise for the system using the bad receiver, see Figs.~\ref{fig:fitness-product-product-pr},~\ref{fig:rec-matrix-pr}, and~\ref{fig:Xhat1-pr}.
Fig.~\ref{fig:toy-approx} shows a plot of the intermediate $K$-approximations necessary to compute the final $K$-approximations, where $K=9000$.

The 2PC transaction managers are depicted in Figs.~\ref{2pc_H},~\ref{2pc_A1}, and~\ref{2pc_A2}.
Due to space, we do not include the intermediate calculations for the 2PC case study, but the plot of the $K$-approximations can be found in Fig.~\ref{fig:2pc-approx}.

The ABP receivers of \cite{AlurTripakisSIGACT17} are depicted
in Figs.~\ref{fig:abp-rec} and \ref{fig:abp-rec_pr} as well as two of their senders in Figs.~\ref{fig:abp-sndr} and \ref{fig:abp-sndr_pr}.
The plot of the $K$-approximations can be found in Fig.~\ref{fig:abp-approx}.

\begin{figure}
    \begin{subfigure}[b]{0.475\textwidth}
        \centering
\begin{center}
\begin{tikzpicture}[scale=\tiksize]
\tikzstyle{every node}+=[inner sep=0pt]
\draw [black] (31.6,-29.4) circle (3);
\draw (31.6,-29.4) node {$p_0^1$};
\draw [black] (31.6,-36.3) circle (3);
\draw (31.6,-36.3) node {$p_0^3$};
\draw [black] (31.6,-36.3) circle (2.4);
\draw [black] (40.7,-36.3) circle (3);
\draw (40.7,-36.3) node {$p_1^2$};
\draw [black] (49.9,-36.3) circle (3);
\draw (49.9,-36.3) node {$p_2^2$};
\draw [black] (22.9,-29.4) -- (28.6,-29.4);
\draw (22.4,-29.4) node [left] {$M\abs{}f_1\mbox{ }=$};
\fill [black] (28.6,-29.4) -- (27.8,-28.9) -- (27.8,-29.9);
\draw [black] (34.581,-29.44) arc (78.13284:27.52532:7.745);
\fill [black] (39.86,-33.44) -- (39.93,-32.5) -- (39.04,-32.96);
\draw (38.56,-30.35) node [above] {$s$};
\draw [black] (37.899,-37.336) arc (-79.04746:-100.95254:9.203);
\fill [black] (34.4,-37.34) -- (35.09,-37.98) -- (35.28,-37);
\draw (36.15,-38) node [below] {$a$};
\draw [black] (34.6,-36.3) -- (37.7,-36.3);
\fill [black] (37.7,-36.3) -- (36.9,-35.8) -- (36.9,-36.8);
\draw (36.15,-35.8) node [above] {$s$};
\draw [black] (43.7,-36.3) -- (46.9,-36.3);
\fill [black] (46.9,-36.3) -- (46.1,-35.8) -- (46.1,-36.8);
\draw (45.3,-35.8) node [above] {$t$};
\draw [black] (47.043,-37.182) arc (-80.76477:-99.23523:10.858);
\fill [black] (43.56,-37.18) -- (44.27,-37.8) -- (44.43,-36.82);
\draw (45.3,-37.82) node [below] {$s$};
\end{tikzpicture}
\end{center}
    \end{subfigure}
    \begin{subfigure}[b]{0.475\textwidth}
        \centering
\begin{center}
\begin{tikzpicture}[scale=\tiksize]
\tikzstyle{every node}+=[inner sep=0pt]
\draw [black] (31.6,-29.4) circle (3);
\draw (31.6,-29.4) node {$p_0^1$};
\draw [black] (31.6,-36.3) circle (3);
\draw (31.6,-36.3) node {$p_0^2$};
\draw [black] (31.6,-36.3) circle (2.4);
\draw [black] (40.7,-36.3) circle (3);
\draw (40.7,-36.3) node {$p_1^2$};
\draw [black] (40.7,-36.3) circle (2.4);
\draw [black] (49.9,-36.3) circle (3);
\draw (49.9,-36.3) node {$p_2^2$};
\draw [black] (49.9,-36.3) circle (2.4);
\draw [black] (22.9,-29.4) -- (28.6,-29.4);
\draw (22.4,-29.4) node [left] {$M\abs{}f_2\mbox{ }=$};
\fill [black] (28.6,-29.4) -- (27.8,-28.9) -- (27.8,-29.9);
\draw [black] (34.581,-29.44) arc (78.13284:27.52532:7.745);
\fill [black] (39.86,-33.44) -- (39.93,-32.5) -- (39.04,-32.96);
\draw (38.56,-30.35) node [above] {$s$};
\draw [black] (37.899,-37.336) arc (-79.04746:-100.95254:9.203);
\fill [black] (34.4,-37.34) -- (35.09,-37.98) -- (35.28,-37);
\draw (36.15,-38) node [below] {$a$};
\draw [black] (34.6,-36.3) -- (37.7,-36.3);
\fill [black] (37.7,-36.3) -- (36.9,-35.8) -- (36.9,-36.8);
\draw (36.15,-35.8) node [above] {$s$};
\draw [black] (43.7,-36.3) -- (46.9,-36.3);
\fill [black] (46.9,-36.3) -- (46.1,-35.8) -- (46.1,-36.8);
\draw (45.3,-35.8) node [above] {$t$};
\draw [black] (47.043,-37.182) arc (-80.76477:-99.23523:10.858);
\fill [black] (43.56,-37.18) -- (44.27,-37.8) -- (44.43,-36.82);
\draw (45.3,-37.82) node [below] {$s$};
\end{tikzpicture}
\end{center}
    \end{subfigure}
    \caption{Synchronous products $M\abs{}f_1$ and $M\abs{}f_2$. In $M\abs{}f_k$, state $p_i^j$ corresponds to the pair of states $p_i$ of $M$ and $f_k^j$ of $f_k$.}
    \label{fig:fitness-product-product}
\end{figure}
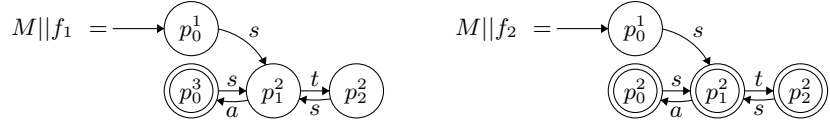

\begin{figure}[!ht]
    \centering
    \input{new_figs/rec-matrix}
    \caption{The recurrence matrices and initial condition vectors $\xi_k,v_k$ correspond to $M\abs{}f_k$ and their intermediate components are
    $\D^{(k)}$, $\A^{(k)}$, $\valpha0^{(k)}$, $\vbeta0^{(k)}$.}
    \label{fig:rec-matrix}
\end{figure}

\begin{figure}[!ht]
    \centering
    \input{new_figs/Xhat1}
    \caption{The symbolic, arithmetic expression for $\Xsum{\img{f}{M_n}}{1}$ and $\Xsum{\img{f}{M_n}}{2}$. When taking a limit, we can ignore the cases where $n = 0$.}
    \label{fig:Xhat1}
\end{figure}

\begin{figure}[!b]
    \begin{subfigure}[b]{0.475\textwidth}
        \centering
\begin{center}
\begin{tikzpicture}[scale=\tiksize]
\tikzstyle{every node}+=[inner sep=0pt]
\draw [black] (31.6,-29.4) circle (3);
\draw (31.6,-29.4) node {${p_0^1}'$};
\draw [black] (40.7,-29.4) circle (3);
\draw (40.7,-29.4) node {${p_1^2}'$};
\draw [black] (49.9,-29.4) circle (3);
\draw (49.9,-29.4) node {${p_2^2}'$};
\draw [black] (31.6,-36.3) circle (3);
\draw (31.6,-36.3) node {${p_0^3}'$};
\draw [black] (31.6,-36.3) circle (2.4);
\draw [black] (40.7,-36.3) circle (3);
\draw (40.7,-36.3) node {${p_3^2}'$};
\draw [black] (49.9,-36.3) circle (3);
\draw (49.9,-36.3) node {${p_4^2}'$};
\draw [black] (33.783,-27.4) arc (116.36909:63.63091:5.329);
\fill [black] (38.52,-27.4) -- (38.02,-26.6) -- (37.58,-27.49);
\draw (36.15,-26.35) node [above] {$s$};
\draw [black] (42.759,-27.275) arc (119.33321:60.66679:5.187);
\fill [black] (47.84,-27.28) -- (47.39,-26.45) -- (46.9,-27.32);
\draw (45.3,-26.11) node [above] {$t$};
\draw [black] (47.5,-31.2) -- (43.1,-34.5);
\fill [black] (43.1,-34.5) -- (44.04,-34.42) -- (43.44,-33.62);
\draw (44.41,-32.35) node [above] {$s$};
\draw [black] (43.7,-36.3) -- (46.9,-36.3);
\fill [black] (46.9,-36.3) -- (46.1,-35.8) -- (46.1,-36.8);
\draw (45.3,-35.8) node [above] {$t$};
\draw [black] (47.064,-37.246) arc (-80.0072:-99.9928:10.168);
\fill [black] (43.54,-37.25) -- (44.24,-37.88) -- (44.41,-36.89);
\draw (45.3,-37.9) node [below] {$s$};
\draw [black] (37.7,-36.3) -- (34.6,-36.3);
\fill [black] (34.6,-36.3) -- (35.4,-36.8) -- (35.4,-35.8);
\draw (36.15,-35.8) node [above] {$a$};
\draw [black] (33.99,-34.49) -- (38.31,-31.21);
\fill [black] (38.31,-31.21) -- (37.37,-31.3) -- (37.97,-32.09);
\draw (35.26,-32.35) node [above] {$s$};
\draw [black] (22.9,-29.4) -- (28.6,-29.4);
\draw (22.4,-29.4) node [left] {$M'\abs{}f_1\mbox{ }=$};
\fill [black] (28.6,-29.4) -- (27.8,-28.9) -- (27.8,-29.9);
\end{tikzpicture}
\end{center}
    \end{subfigure}
    \begin{subfigure}[b]{0.475\textwidth}
        \centering
\begin{center}
\begin{tikzpicture}[scale=\tiksize]
\tikzstyle{every node}+=[inner sep=0pt]
\draw [black] (31.6,-29.4) circle (3);
\draw (31.6,-29.4) node {${p_0^1}'$};
\draw [black] (40.7,-29.4) circle (3);
\draw (40.7,-29.4) node {${p_1^2}'$};
\draw [black] (40.7,-29.4) circle (2.4);
\draw [black] (49.9,-29.4) circle (3);
\draw (49.9,-29.4) node {${p_2^2}'$};
\draw [black] (49.9,-29.4) circle (2.4);
\draw [black] (31.6,-36.3) circle (3);
\draw (31.6,-36.3) node {${p_0^2}'$};
\draw [black] (31.6,-36.3) circle (2.4);
\draw [black] (40.7,-36.3) circle (3);
\draw (40.7,-36.3) node {${p_3^2}'$};
\draw [black] (40.7,-36.3) circle (2.4);
\draw [black] (49.9,-36.3) circle (3);
\draw (49.9,-36.3) node {${p_4^2}'$};
\draw [black] (49.9,-36.3) circle (2.4);
\draw [black] (33.783,-27.4) arc (116.36909:63.63091:5.329);
\fill [black] (38.52,-27.4) -- (38.02,-26.6) -- (37.58,-27.49);
\draw (36.15,-26.35) node [above] {$s$};
\draw [black] (42.759,-27.275) arc (119.33321:60.66679:5.187);
\fill [black] (47.84,-27.28) -- (47.39,-26.45) -- (46.9,-27.32);
\draw (45.3,-26.11) node [above] {$t$};
\draw [black] (47.5,-31.2) -- (43.1,-34.5);
\fill [black] (43.1,-34.5) -- (44.04,-34.42) -- (43.44,-33.62);
\draw (44.41,-32.35) node [above] {$s$};
\draw [black] (43.7,-36.3) -- (46.9,-36.3);
\fill [black] (46.9,-36.3) -- (46.1,-35.8) -- (46.1,-36.8);
\draw (45.3,-35.8) node [above] {$t$};
\draw [black] (47.064,-37.246) arc (-80.0072:-99.9928:10.168);
\fill [black] (43.54,-37.25) -- (44.24,-37.88) -- (44.41,-36.89);
\draw (45.3,-37.9) node [below] {$s$};
\draw [black] (37.7,-36.3) -- (34.6,-36.3);
\fill [black] (34.6,-36.3) -- (35.4,-36.8) -- (35.4,-35.8);
\draw (36.15,-35.8) node [above] {$a$};
\draw [black] (33.99,-34.49) -- (38.31,-31.21);
\fill [black] (38.31,-31.21) -- (37.37,-31.3) -- (37.97,-32.09);
\draw (35.26,-32.35) node [above] {$s$};
\draw [black] (22.9,-29.4) -- (28.6,-29.4);
\draw (22.4,-29.4) node [left] {$M'\abs{}f_2\mbox{ }=$};
\fill [black] (28.6,-29.4) -- (27.8,-28.9) -- (27.8,-29.9);
\end{tikzpicture}
\end{center}
    \end{subfigure}
    \caption{Synchronous products $M'\abs{}f_1$ and $M'\abs{}f_2$. In $M'\abs{}f_k$, state ${p_i^j}'$ corresponds to the pair of states $p_i'$ of $M'$ and $f_k^j$ of $f_k$.}
    \label{fig:fitness-product-product-pr}
\end{figure}
\begin{figure}
    \centering
    \input{new_figs/rec-matrix-pr}
    \caption{The recurrence matrices and initial condition vectors $\xi_k',v_k'$ correspond to $M'\abs{}f_k$ and their intermediate components are
    $\D'^{(k)}$, $\A'^{(k)}$, $\valpha0'^{(k)}$, $\vbeta0'^{(k)}$.}
    \label{fig:rec-matrix-pr}
\end{figure}
\begin{figure}[!b]
    \centering
    \input{new_figs/Xhat1-pr}
    \caption{The symbolic, arithmetic expression for $\Xsum{\img{f}{M'_n}}{1}$. This expression can be represented as a non-piecewise function using complex numbers; indeed that is how Mathematica expresses it by default. The terms $c_j, k_j$ are $\sqrt{5} + j$ and $\sqrt{5} - j$ respectively. The term $m := \floor{n/4}$. We omit the expression for $\Xsum{\img{f}{M'_n}}{2}$.}
    \label{fig:Xhat1-pr}
\end{figure}
\begin{figure}[!ht]
    \centering
    \includegraphics[width=\textwidth]{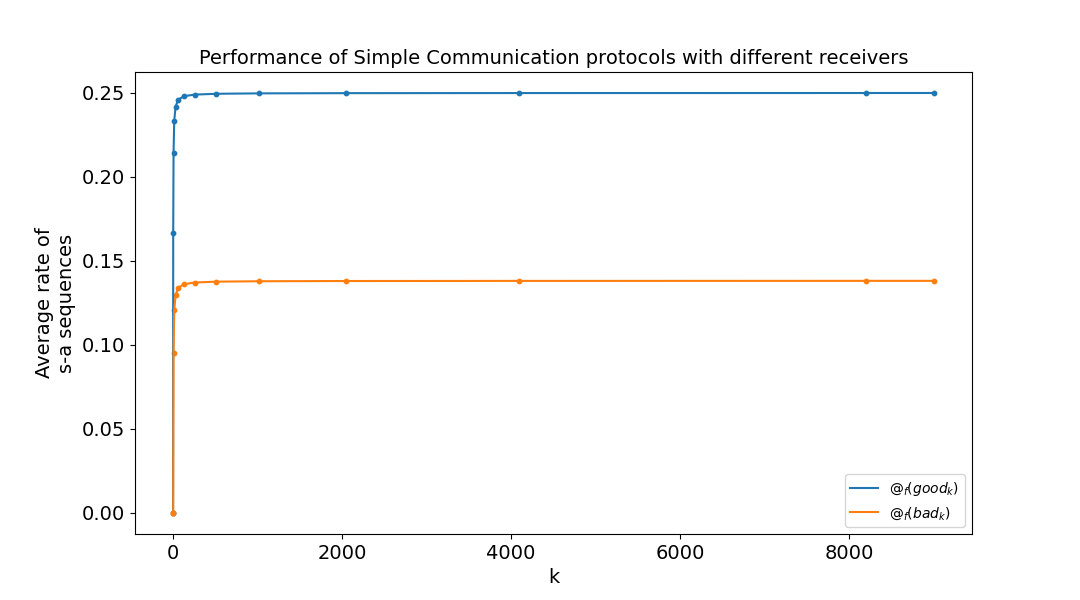}
    \caption{A graph of the $K$-approximation against $K$ for the simple communication protocol case study.}
    \label{fig:toy-approx}
\end{figure}

\begin{figure}
    \centering
\begin{center}
\begin{tikzpicture}[scale=\tiksize]
\tikzstyle{every node}+=[inner sep=0pt]
\draw [black] (9.2,-14.6) circle (3);
\draw (9.2,-14.6) node {$m_0$};
\draw [black] (21.4,-14.6) circle (3);
\draw (21.4,-14.6) node {$m_1$};
\draw [black] (33.6,-14.6) circle (3);
\draw (33.6,-14.6) node {$m_2$};
\draw [black] (45.8,-14.6) circle (3);
\draw (45.8,-14.6) node {$m_3$};
\draw [black] (39.7,-26.8) circle (3);
\draw (39.7,-26.8) node {$m_5$};
\draw [black] (51.9,-26.8) circle (3);
\draw (51.9,-26.8) node {$m_4$};
\draw [black] (27.5,-26.8) circle (3);
\draw (27.5,-26.8) node {$m_6$};
\draw [black] (15.3,-26.8) circle (3);
\draw (15.3,-26.8) node {$m_7$};
\draw [black] (9.2,-38.9) circle (3);
\draw (9.2,-38.9) node {$m_{11}$};
\draw [black] (21.4,-38.9) circle (3);
\draw (21.4,-38.9) node {$m_{10}$};
\draw [black] (33.6,-38.9) circle (3);
\draw (33.6,-38.9) node {$m_9$};
\draw [black] (45.8,-38.9) circle (3);
\draw (45.8,-38.9) node {$m_8$};
\draw [black] (4.1,-14.6) -- (6.2,-14.6);
\draw (3.6,-14.6) node [left] {$M_H$};
\fill [black] (6.2,-14.6) -- (5.4,-14.1) -- (5.4,-15.1);
\draw [black] (12.2,-14.6) -- (18.4,-14.6);
\fill [black] (18.4,-14.6) -- (17.6,-14.1) -- (17.6,-15.1);
\draw (15.3,-14.1) node [above] {$x?$};
\draw [black] (24.4,-14.6) -- (30.6,-14.6);
\fill [black] (30.6,-14.6) -- (29.8,-14.1) -- (29.8,-15.1);
\draw (27.5,-14.1) node [above] {$x_1!$};
\draw [black] (36.6,-14.6) -- (42.8,-14.6);
\fill [black] (42.8,-14.6) -- (42,-14.1) -- (42,-15.1);
\draw (39.7,-14.1) node [above] {$x_2!$};
\draw [black] (48.771,-14.238) arc (90.0819:-90.0819:12.512);
\fill [black] (48.77,-39.26) -- (49.57,-39.76) -- (49.57,-38.76);
\draw (61.8,-26.75) node [right] {$\textit{no}_*?$};
\draw [black] (47.14,-17.28) -- (50.56,-24.12);
\fill [black] (50.56,-24.12) -- (50.65,-23.18) -- (49.75,-23.62);
\draw (48.15,-21.81) node [left] {$\textit{yes}_*?$};
\draw [black] (50.55,-29.48) -- (47.15,-36.22);
\fill [black] (47.15,-36.22) -- (47.96,-35.73) -- (47.06,-35.28);
\draw [black] (48.9,-26.8) -- (42.7,-26.8);
\fill [black] (42.7,-26.8) -- (43.5,-27.3) -- (43.5,-26.3);
\draw (45.8,-26.3) node [above] {$\textit{yes}_*?$};
\draw [black] (36.7,-26.8) -- (30.5,-26.8);
\fill [black] (30.5,-26.8) -- (31.3,-27.3) -- (31.3,-26.3);
\draw (33.6,-26.3) node [above] {$\textit{cm}_1!$};
\draw [black] (24.5,-26.8) -- (18.3,-26.8);
\fill [black] (18.3,-26.8) -- (19.1,-27.3) -- (19.1,-26.3);
\draw (21.4,-26.3) node [above] {$\textit{cm}_2!$};
\draw [black] (13.96,-24.12) -- (10.54,-17.28);
\fill [black] (10.54,-17.28) -- (10.45,-18.22) -- (11.35,-17.78);
\draw (12.95,-19.59) node [right] {$\textit{succ}!$};
\draw [black] (9.2,-35.9) -- (9.2,-17.6);
\fill [black] (9.2,-17.6) -- (8.7,-18.4) -- (9.7,-18.4);
\draw (8.7,-26.75) node [left] {$\textit{fail}!$};
\draw [black] (18.4,-38.9) -- (12.2,-38.9);
\fill [black] (12.2,-38.9) -- (13,-39.4) -- (13,-38.4);
\draw (15.3,-38.4) node [above] {$\textit{ab}_2!$};
\draw [black] (30.6,-38.9) -- (24.4,-38.9);
\fill [black] (24.4,-38.9) -- (25.2,-39.4) -- (25.2,-38.4);
\draw (27.5,-38.4) node [above] {$\textit{ab}_1!$};
\draw [black] (42.8,-38.9) -- (36.6,-38.9);
\fill [black] (36.6,-38.9) -- (37.4,-39.4) -- (37.4,-38.4);
\draw (39.7,-38.4) node [above] {$\textit{yes}_*?$};
\draw (39.7,-36.4) node [above] {$\textit{no}_*?$};
\draw [black] (46.629,-11.729) arc (191.63248:-96.36752:2.25);
\draw (51.81,-9.13) node [above] {$x?$};
\fill [black] (48.58,-13.51) -- (49.47,-13.84) -- (49.27,-12.86);
\draw [black] (7.877,-11.92) arc (234:-54:2.25);
\draw (9.2,-7.35) node [above] {$\textit{yes}_*?,\textit{no}_*?$};
\fill [black] (10.52,-11.92) -- (11.4,-11.57) -- (10.59,-10.98);
\draw [black] (54.257,-28.637) arc (79.80895:-208.19105:2.25);
\draw (56.27,-33.39) node [below] {$x?$};
\fill [black] (51.88,-29.79) -- (51.24,-30.49) -- (52.23,-30.66);
\end{tikzpicture}
\end{center}
    \caption{The transaction manager used by the 2PC system labeled H. The label $\textit{yes}_*$ denotes both $\textit{yes}_1$ and $\textit{yes}_2$. Likewise for $\textit{no}_*$.}
    \label{2pc_H}
\end{figure}
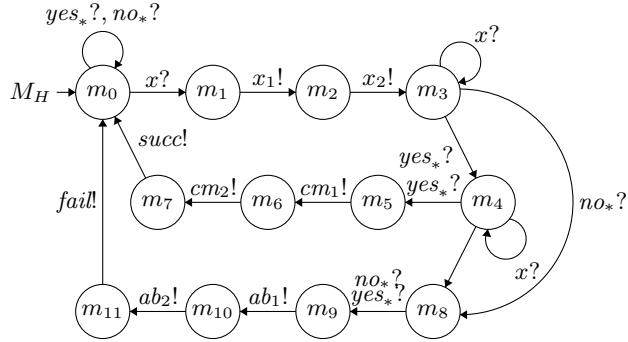
\begin{figure}
    \centering
\begin{center}
\begin{tikzpicture}[scale=\tiksize]
\tikzstyle{every node}+=[inner sep=0pt]
\draw [black] (9.2,-14.6) circle (3);
\draw (9.2,-14.6) node {$m_0$};
\draw [black] (21.4,-14.6) circle (3);
\draw (21.4,-14.6) node {$m_1$};
\draw [black] (39.7,-14.6) circle (3);
\draw (39.7,-14.6) node {$m_2$};
\draw [black] (54.8,-26.8) circle (3);
\draw (54.8,-26.8) node {$m_3$};
\draw [black] (39.7,-26.8) circle (3);
\draw (39.7,-26.8) node {$m_5$};
\draw [black] (52,-14.6) circle (3);
\draw (52,-14.6) node {$m_4$};
\draw [black] (27.5,-26.8) circle (3);
\draw (27.5,-26.8) node {$m_6$};
\draw [black] (15.3,-26.8) circle (3);
\draw (15.3,-26.8) node {$m_7$};
\draw [black] (9.2,-38.9) circle (3);
\draw (9.2,-38.9) node {$m_{11}$};
\draw [black] (21.4,-38.9) circle (3);
\draw (21.4,-38.9) node {$m_{10}$};
\draw [black] (33.6,-38.9) circle (3);
\draw (33.6,-38.9) node {$m_9$};
\draw [black] (45.8,-38.9) circle (3);
\draw (45.8,-38.9) node {$m_8$};
\draw [black] (4.1,-14.6) -- (6.2,-14.6);
\draw (3.6,-14.6) node [left] {$M_A$};
\fill [black] (6.2,-14.6) -- (5.4,-14.1) -- (5.4,-15.1);
\draw [black] (12.2,-14.6) -- (18.4,-14.6);
\fill [black] (18.4,-14.6) -- (17.6,-14.1) -- (17.6,-15.1);
\draw (15.3,-14.1) node [above] {$x?$};
\draw [black] (24.4,-14.6) -- (36.7,-14.6);
\fill [black] (36.7,-14.6) -- (35.9,-14.1) -- (35.9,-15.1);
\draw (30.55,-14.1) node [above] {$x_1!$};
\draw [black] (36.7,-26.8) -- (30.5,-26.8);
\fill [black] (30.5,-26.8) -- (31.3,-27.3) -- (31.3,-26.3);
\draw (33.6,-26.3) node [above] {$\textit{cm}_1!$};
\draw [black] (24.5,-26.8) -- (18.3,-26.8);
\fill [black] (18.3,-26.8) -- (19.1,-27.3) -- (19.1,-26.3);
\draw (21.4,-26.3) node [above] {$\textit{cm}_2!$};
\draw [black] (13.96,-24.12) -- (10.54,-17.28);
\fill [black] (10.54,-17.28) -- (10.45,-18.22) -- (11.35,-17.78);
\draw (12.95,-19.59) node [right] {$\textit{succ}!$};
\draw [black] (9.2,-35.9) -- (9.2,-17.6);
\fill [black] (9.2,-17.6) -- (8.7,-18.4) -- (9.7,-18.4);
\draw (8.7,-26.75) node [left] {$\textit{fail}!$};
\draw [black] (18.4,-38.9) -- (12.2,-38.9);
\fill [black] (12.2,-38.9) -- (13,-39.4) -- (13,-38.4);
\draw (15.3,-38.4) node [above] {$\textit{ab}_2!$};
\draw [black] (30.6,-38.9) -- (24.4,-38.9);
\fill [black] (24.4,-38.9) -- (25.2,-39.4) -- (25.2,-38.4);
\draw (27.5,-38.4) node [above] {$\textit{ab}_1!$};
\draw [black] (7.877,-11.92) arc (234:-54:2.25);
\draw (9.2,-7.35) node [above] {$\textit{yes}_*?,\textit{no}_*?$};
\fill [black] (10.52,-11.92) -- (11.4,-11.57) -- (10.59,-10.98);
\draw [black] (36.821,-13.8) arc (282.20652:-5.79348:2.25);
\draw (34.13,-8.65) node [left] {$x?$};
\fill [black] (38.58,-11.83) -- (38.9,-10.94) -- (37.93,-11.15);
\draw [black] (42.7,-14.6) -- (49,-14.6);
\fill [black] (49,-14.6) -- (48.2,-14.1) -- (48.2,-15.1);
\draw (45.85,-14.1) node [above] {$\textit{yes}_1?$};
\draw [black] (42.03,-16.49) -- (52.47,-24.91);
\fill [black] (52.47,-24.91) -- (52.16,-24.02) -- (51.53,-24.8);
\draw (49.5,-20.21) node [above] {$\textit{no}_1?$};
\draw [black] (39.7,-17.6) -- (39.7,-23.8);
\fill [black] (39.7,-23.8) -- (40.2,-23) -- (39.2,-23);
\draw (39.2,-20.7) node [left] {$\textit{yes}_2?$};
\draw [black] (41.621,-16.897) arc (33.96341:-62.14683:14.475);
\fill [black] (36.38,-37.78) -- (37.32,-37.85) -- (36.85,-36.97);
\draw (44.42,-28.97) node [right] {$\textit{no}_2?$};
\draw [black] (53.01,-29.21) -- (47.59,-36.49);
\fill [black] (47.59,-36.49) -- (48.47,-36.15) -- (47.67,-35.55);
\draw (50.88,-34.24) node [right] {$x_2!$};
\draw [black] (23.06,-12.106) arc (141.42372:38.57628:17.447);
\fill [black] (23.06,-12.11) -- (23.95,-11.79) -- (23.17,-11.17);
\draw (36.7,-5.04) node [above] {$x_2!$};
\draw [black] (42.8,-38.9) -- (36.6,-38.9);
\fill [black] (36.6,-38.9) -- (37.4,-39.4) -- (37.4,-38.4);
\draw (39.7,-39.4) node [below] {$\textit{yes}_*?$};
\draw (39.7,-41.4) node [below] {$\textit{no}_*?$};
\draw [black] (47.123,-41.58) arc (54:-234:2.25);
\draw (45.8,-46.15) node [below] {$x?$};
\fill [black] (44.48,-41.58) -- (43.6,-41.93) -- (44.41,-42.52);
\end{tikzpicture}
\end{center}
    \caption{The transaction manager used by the 2PC system labeled A1. The label $\textit{yes}_*$ denotes both $\textit{yes}_1$ and $\textit{yes}_2$. Likewise for $\textit{no}_*$.}
    \label{2pc_A1}
\end{figure}
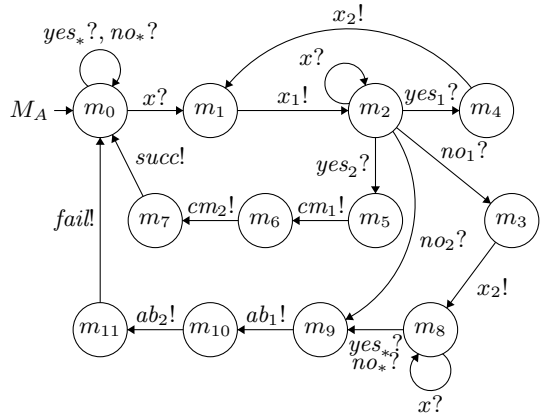
\begin{figure}
    \vspace{1em}
    \centering
\begin{center}
\begin{tikzpicture}[scale=\tiksize]
\tikzstyle{every node}+=[inner sep=0pt]
\draw [black] (9.2,-14.6) circle (3);
\draw (9.2,-14.6) node {$m_0$};
\draw [black] (21.4,-14.6) circle (3);
\draw (21.4,-14.6) node {$m_1$};
\draw [black] (39.7,-14.6) circle (3);
\draw (39.7,-14.6) node {$m_2$};
\draw [black] (54.8,-26.8) circle (3);
\draw (54.8,-26.8) node {$m_3$};
\draw [black] (39.7,-26.8) circle (3);
\draw (39.7,-26.8) node {$m_5$};
\draw [black] (52,-14.6) circle (3);
\draw (52,-14.6) node {$m_4$};
\draw [black] (27.5,-26.8) circle (3);
\draw (27.5,-26.8) node {$m_6$};
\draw [black] (15.3,-26.8) circle (3);
\draw (15.3,-26.8) node {$m_7$};
\draw [black] (9.2,-38.9) circle (3);
\draw (9.2,-38.9) node {$m_{11}$};
\draw [black] (21.4,-38.9) circle (3);
\draw (21.4,-38.9) node {$m_{10}$};
\draw [black] (33.6,-38.9) circle (3);
\draw (33.6,-38.9) node {$m_9$};
\draw [black] (45.8,-38.9) circle (3);
\draw (45.8,-38.9) node {$m_8$};
\draw [black] (4.1,-14.6) -- (6.2,-14.6);
\draw (3.6,-14.6) node [left] {$M_A$};
\fill [black] (6.2,-14.6) -- (5.4,-14.1) -- (5.4,-15.1);
\draw [black] (12.2,-14.6) -- (18.4,-14.6);
\fill [black] (18.4,-14.6) -- (17.6,-14.1) -- (17.6,-15.1);
\draw (15.3,-14.1) node [above] {$x?$};
\draw [black] (24.4,-14.6) -- (36.7,-14.6);
\fill [black] (36.7,-14.6) -- (35.9,-14.1) -- (35.9,-15.1);
\draw (30.55,-14.1) node [above] {$x_1!$};
\draw [black] (36.7,-26.8) -- (30.5,-26.8);
\fill [black] (30.5,-26.8) -- (31.3,-27.3) -- (31.3,-26.3);
\draw (33.6,-26.3) node [above] {$\textit{cm}_1!$};
\draw [black] (24.5,-26.8) -- (18.3,-26.8);
\fill [black] (18.3,-26.8) -- (19.1,-27.3) -- (19.1,-26.3);
\draw (21.4,-26.3) node [above] {$\textit{cm}_2!$};
\draw [black] (13.96,-24.12) -- (10.54,-17.28);
\fill [black] (10.54,-17.28) -- (10.45,-18.22) -- (11.35,-17.78);
\draw (12.95,-19.59) node [right] {$\textit{succ}!$};
\draw [black] (9.2,-35.9) -- (9.2,-17.6);
\fill [black] (9.2,-17.6) -- (8.7,-18.4) -- (9.7,-18.4);
\draw (8.7,-26.75) node [left] {$\textit{fail}!$};
\draw [black] (18.4,-38.9) -- (12.2,-38.9);
\fill [black] (12.2,-38.9) -- (13,-39.4) -- (13,-38.4);
\draw (15.3,-38.4) node [above] {$\textit{ab}_2!$};
\draw [black] (30.6,-38.9) -- (24.4,-38.9);
\fill [black] (24.4,-38.9) -- (25.2,-39.4) -- (25.2,-38.4);
\draw (27.5,-38.4) node [above] {$\textit{ab}_1!$};
\draw [black] (7.877,-11.92) arc (234:-54:2.25);
\draw (9.2,-7.35) node [above] {$\textit{yes}_*?,\textit{no}_*?$};
\fill [black] (10.52,-11.92) -- (11.4,-11.57) -- (10.59,-10.98);
\draw [black] (36.821,-13.8) arc (282.20652:-5.79348:2.25);
\draw (34.13,-8.65) node [left] {$x?$};
\fill [black] (38.58,-11.83) -- (38.9,-10.94) -- (37.93,-11.15);
\draw [black] (42.7,-14.6) -- (49,-14.6);
\fill [black] (49,-14.6) -- (48.2,-14.1) -- (48.2,-15.1);
\draw (45.85,-14.1) node [above] {$\textit{yes}_1?$};
\draw [black] (42.03,-16.49) -- (52.47,-24.91);
\fill [black] (52.47,-24.91) -- (52.16,-24.02) -- (51.53,-24.8);
\draw (49.5,-20.21) node [above] {$\textit{no}_1?$};
\draw [black] (39.7,-17.6) -- (39.7,-23.8);
\fill [black] (39.7,-23.8) -- (40.2,-23) -- (39.2,-23);
\draw (39.2,-20.7) node [left] {$\textit{yes}_2?$};
\draw [black] (41.621,-16.897) arc (33.96341:-62.14683:14.475);
\fill [black] (36.38,-37.78) -- (37.32,-37.85) -- (36.85,-36.97);
\draw (44.42,-28.97) node [right] {$\textit{no}_2?$};
\draw [black] (53.01,-29.21) -- (47.59,-36.49);
\fill [black] (47.59,-36.49) -- (48.47,-36.15) -- (47.67,-35.55);
\draw (50.88,-34.24) node [right] {$x_2!$};
\draw [black] (42.8,-38.9) -- (36.6,-38.9);
\fill [black] (36.6,-38.9) -- (37.4,-39.4) -- (37.4,-38.4);
\draw (39.7,-39.4) node [below] {$\textit{yes}_*?$};
\draw (39.7,-41.4) node [below] {$\textit{no}_*?$};
\draw [black] (47.123,-41.58) arc (54:-234:2.25);
\draw (45.8,-46.15) node [below] {$x?$};
\fill [black] (44.48,-41.58) -- (43.6,-41.93) -- (44.41,-42.52);
\draw [black] (41.075,-11.965) arc (138.88405:41.11595:6.338);
\fill [black] (41.07,-11.97) -- (41.98,-11.69) -- (41.22,-11.03);
\draw (45.85,-9.29) node [above] {$x_2!$};
\end{tikzpicture}
\end{center}
    \caption{The transaction manager used by the 2PC system labeled A2. The label $\textit{yes}_*$ denotes both $\textit{yes}_1$ and $\textit{yes}_2$. Likewise for $\textit{no}_*$.}
    \label{2pc_A2}
\end{figure}
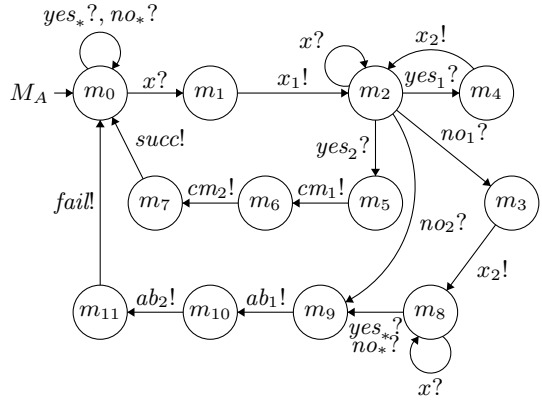
\begin{figure}[!ht]
    \centering
    \includegraphics[width=\textwidth]{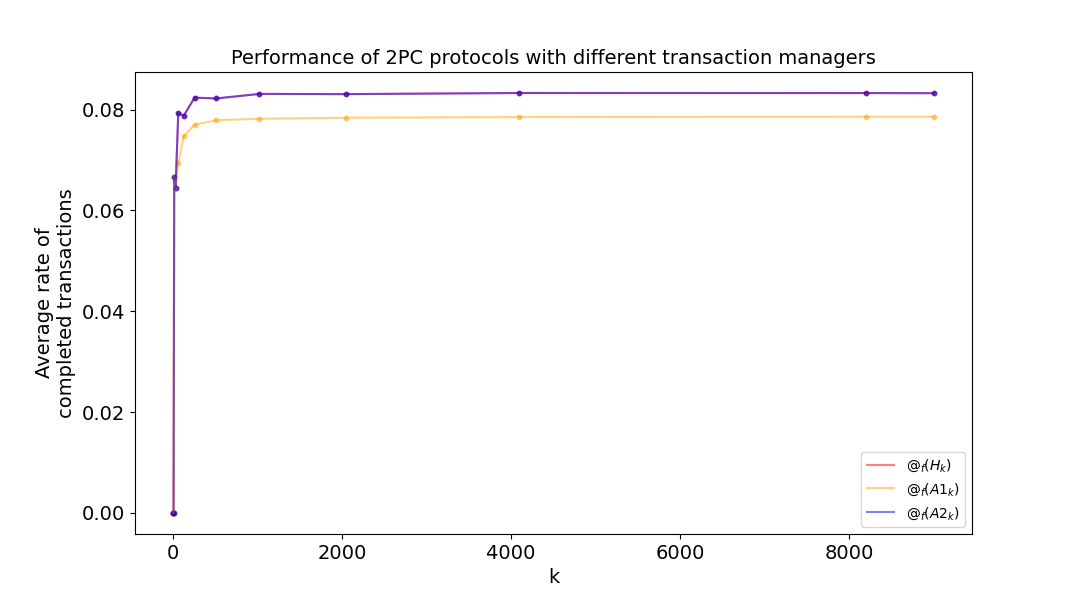}
    \caption{A graph of the $K$-approximation against $K$ for the 2PC case study. Note: $\agg_f(\text{H}_k)$ and $\agg_f(\text{A2}_k)$ overlap completely.}
    \label{fig:2pc-approx}
\end{figure}

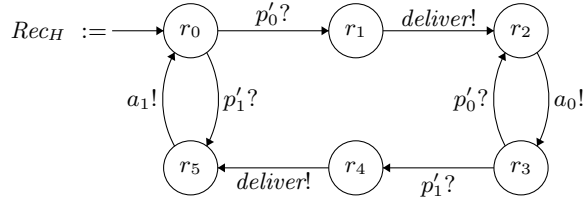
\begin{figure}[!ht]
    \vspace{1em}
    \centering

\begin{center}
\begin{tikzpicture}[scale=\tiksize]
\tikzstyle{every node}+=[inner sep=0pt]
\draw [black] (22.5,-33.2) circle (3);
\draw (22.5,-33.2) node {$r_0$};
\draw [black] (40.8,-33.2) circle (3);
\draw (40.8,-33.2) node {$r_1$};
\draw [black] (59.1,-33.2) circle (3);
\draw (59.1,-33.2) node {$r_2$};
\draw [black] (22.5,-48.4) circle (3);
\draw (22.5,-48.4) node {$r_5$};
\draw [black] (40.8,-48.4) circle (3);
\draw (40.8,-48.4) node {$r_4$};
\draw [black] (59.1,-48.4) circle (3);
\draw (59.1,-48.4) node {$r_3$};
\draw [black] (25.5,-33.2) -- (37.8,-33.2);
\fill [black] (37.8,-33.2) -- (37,-32.7) -- (37,-33.7);
\draw (31.65,-32.7) node [above] {$p'_0?$};
\draw [black] (43.8,-33.2) -- (56.1,-33.2);
\fill [black] (56.1,-33.2) -- (55.3,-32.7) -- (55.3,-33.7);
\draw (49.95,-32.7) node [above] {$\textit{deliver}!$};
\draw [black] (60.871,-35.61) arc (28.43666:-28.43666:10.9);
\fill [black] (60.87,-45.99) -- (61.69,-45.53) -- (60.81,-45.05);
\draw (62.69,-40.8) node [right] {$a_0!$};
\draw [black] (20.721,-45.996) arc (-151.41152:-208.58848:10.859);
\fill [black] (20.72,-35.6) -- (19.9,-36.07) -- (20.78,-36.55);
\draw (18.9,-40.8) node [left] {$a_1!$};
\draw [black] (56.1,-48.4) -- (43.8,-48.4);
\fill [black] (43.8,-48.4) -- (44.6,-48.9) -- (44.6,-47.9);
\draw (49.95,-48.9) node [below] {$p'_1?$};
\draw [black] (24.271,-35.61) arc (28.43666:-28.43666:10.9);
\fill [black] (24.27,-45.99) -- (25.09,-45.53) -- (24.21,-45.05);
\draw (26.09,-40.8) node [right] {$p'_1?$};
\draw [black] (57.353,-45.973) arc (-152.04582:-207.95418:11.034);
\fill [black] (57.35,-35.63) -- (56.54,-36.1) -- (57.42,-36.57);
\draw (55.57,-40.8) node [left] {$p'_0?$};
\draw [black] (13.8,-33.2) -- (19.5,-33.2);
\draw (13.3,-33.2) node [left] {$\textit{Rec}_H\mbox{ }:=$};
\fill [black] (19.5,-33.2) -- (18.7,-32.7) -- (18.7,-33.7);
\draw [black] (37.8,-48.4) -- (25.5,-48.4);
\fill [black] (25.5,-48.4) -- (26.3,-48.9) -- (26.3,-47.9);
\draw (31.65,-48.9) node [below] {$\textit{deliver}!$};
\end{tikzpicture}
\end{center}
    \caption{A manually constructed ABP receiver from~\cite{AlurTripakisSIGACT17}. Note: the synthesis algorithm of~\cite{AlurTripakisSIGACT17} was able to automatically synthesize this receiver, but we will refer to it as the {\it human-made} receiver.}
    \label{fig:abp-rec}
\end{figure}
\begin{figure}[!ht]
    \centering

\begin{center}
\begin{tikzpicture}[scale=\tiksize]
\tikzstyle{every node}+=[inner sep=0pt]
\draw [black] (22.5,-33.2) circle (3);
\draw (22.5,-33.2) node {$r_0$};
\draw [black] (40.8,-33.2) circle (3);
\draw (40.8,-33.2) node {$r_1$};
\draw [black] (59.1,-33.2) circle (3);
\draw (59.1,-33.2) node {$r_2$};
\draw [black] (22.5,-48.4) circle (3);
\draw (22.5,-48.4) node {$r_5$};
\draw [black] (40.8,-48.4) circle (3);
\draw (40.8,-48.4) node {$r_4$};
\draw [black] (59.1,-48.4) circle (3);
\draw (59.1,-48.4) node {$r_3$};
\draw [black] (25.5,-33.2) -- (37.8,-33.2);
\fill [black] (37.8,-33.2) -- (37,-32.7) -- (37,-33.7);
\draw (31.65,-32.7) node [above] {$p'_0?$};
\draw [black] (43.8,-33.2) -- (56.1,-33.2);
\fill [black] (56.1,-33.2) -- (55.3,-32.7) -- (55.3,-33.7);
\draw (49.95,-32.7) node [above] {$\textit{deliver}!$};
\draw [black] (60.871,-35.61) arc (28.43666:-28.43666:10.9);
\fill [black] (60.87,-45.99) -- (61.69,-45.53) -- (60.81,-45.05);
\draw (62.69,-40.8) node [right] {$a_0!$};
\draw [black] (20.721,-45.996) arc (-151.41152:-208.58848:10.859);
\fill [black] (20.72,-35.6) -- (19.9,-36.07) -- (20.78,-36.55);
\draw (18.9,-40.8) node [left] {$a_1!$};
\draw [black] (56.1,-48.4) -- (43.8,-48.4);
\fill [black] (43.8,-48.4) -- (44.6,-48.9) -- (44.6,-47.9);
\draw (49.95,-48.9) node [below] {$p'_1?$};
\draw [black] (38.49,-46.48) -- (24.81,-35.12);
\fill [black] (24.81,-35.12) -- (25.1,-36.01) -- (25.74,-35.24);
\draw (35.27,-40.31) node [above] {$\textit{deliver}!$};
\draw [black] (24.271,-35.61) arc (28.43666:-28.43666:10.9);
\fill [black] (24.27,-45.99) -- (25.09,-45.53) -- (24.21,-45.05);
\draw (26.09,-40.8) node [right] {$p'_1?$};
\draw [black] (57.353,-45.973) arc (-152.04582:-207.95418:11.034);
\fill [black] (57.35,-35.63) -- (56.54,-36.1) -- (57.42,-36.57);
\draw (55.57,-40.8) node [left] {$p'_0?$};
\draw [black] (13.8,-33.2) -- (19.5,-33.2);
\draw (13.3,-33.2) node [left] {$\textit{Rec}_A\mbox{ }:=$};
\fill [black] (19.5,-33.2) -- (18.7,-32.7) -- (18.7,-33.7);
\end{tikzpicture}
\end{center}
    \caption{An automatically synthesized ABP receiver from~\cite{AlurTripakisSIGACT17}. We will refer to this receiver as the {\it algorithm-made} receiver.}
    \label{fig:abp-rec_pr}
\end{figure}
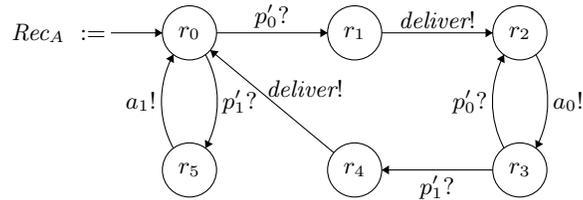
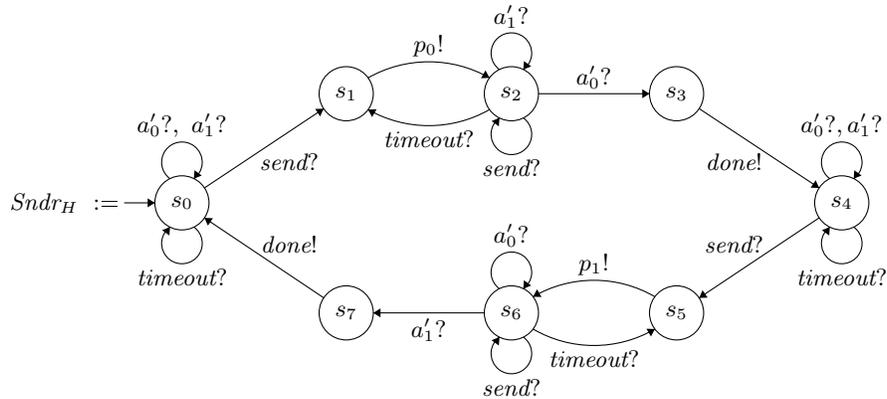
\begin{figure}[!ht]
    \centering

\begin{center}
\begin{tikzpicture}[scale=\tiksize]
\tikzstyle{every node}+=[inner sep=0pt]
\draw [black] (3.8,-29.2) circle (3);
\draw (3.8,-29.2) node {$s_0$};
\draw [black] (22,-17.1) circle (3);
\draw (22,-17.1) node {$s_1$};
\draw [black] (40.3,-17.1) circle (3);
\draw (40.3,-17.1) node {$s_2$};
\draw [black] (58.6,-17.1) circle (3);
\draw (58.6,-17.1) node {$s_3$};
\draw [black] (76.9,-29.2) circle (3);
\draw (76.9,-29.2) node {$s_4$};
\draw [black] (22,-41.4) circle (3);
\draw (22,-41.4) node {$s_7$};
\draw [black] (40.3,-41.4) circle (3);
\draw (40.3,-41.4) node {$s_6$};
\draw [black] (58.6,-41.4) circle (3);
\draw (58.6,-41.4) node {$s_5$};
\draw [black] (2.477,-26.52) arc (234:-54:2.25);
\draw (3.8,-21.95) node [above] {$a'_0?,\mbox{ }a'_1?$};
\fill [black] (5.12,-26.52) -- (6,-26.17) -- (5.19,-25.58);
\draw [black] (6.3,-27.54) -- (19.5,-18.76);
\fill [black] (19.5,-18.76) -- (18.56,-18.79) -- (19.11,-19.62);
\draw (15.68,-23.65) node [below] {$\textit{send}?$};
\draw [black] (24.423,-15.341) arc (119.66327:60.33673:13.593);
\fill [black] (37.88,-15.34) -- (37.43,-14.51) -- (36.93,-15.38);
\draw (31.15,-13.06) node [above] {$p_0!$};
\draw [black] (43.3,-17.1) -- (55.6,-17.1);
\fill [black] (55.6,-17.1) -- (54.8,-16.6) -- (54.8,-17.6);
\draw (49.45,-16.6) node [above] {$a'_0?$};
\draw [black] (61.1,-18.75) -- (74.4,-27.55);
\fill [black] (74.4,-27.55) -- (74.01,-26.69) -- (73.45,-27.52);
\draw (64.97,-23.65) node [below] {$\textit{done}!$};
\draw [black] (74.4,-30.86) -- (61.1,-39.74);
\fill [black] (61.1,-39.74) -- (62.04,-39.71) -- (61.48,-38.88);
\draw (64.97,-34.8) node [above] {$\textit{send}?$};
\draw [black] (42.7,-39.611) arc (120.28318:59.71682:13.385);
\fill [black] (42.7,-39.61) -- (43.64,-39.64) -- (43.14,-38.78);
\draw (49.45,-37.28) node [above] {$p_1!$};
\draw [black] (37.3,-41.4) -- (25,-41.4);
\fill [black] (25,-41.4) -- (25.8,-41.9) -- (25.8,-40.9);
\draw (31.15,-41.9) node [below] {$a'_1?$};
\draw [black] (19.51,-39.73) -- (6.29,-30.87);
\fill [black] (6.29,-30.87) -- (6.68,-31.73) -- (7.23,-30.9);
\draw (15.68,-34.8) node [above] {$\textit{done}!$};
\draw [black] (37.877,-18.859) arc (-60.33673:-119.66327:13.593);
\fill [black] (24.42,-18.86) -- (24.87,-19.69) -- (25.37,-18.82);
\draw (31.15,-21.14) node [below] {$\textit{timeout}?$};
\draw [black] (56.2,-43.189) arc (-59.71682:-120.28318:13.385);
\fill [black] (56.2,-43.19) -- (55.26,-43.16) -- (55.76,-44.02);
\draw (49.45,-45.52) node [below] {$\textit{timeout}?$};
\draw [black] (5.123,-31.88) arc (54:-234:2.25);
\draw (3.8,-36.45) node [below] {$\textit{timeout}?$};
\fill [black] (2.48,-31.88) -- (1.6,-32.23) -- (2.41,-32.82);
\draw [black] (41.623,-44.08) arc (54:-234:2.25);
\draw (40.3,-48.65) node [below] {$\textit{send}?$};
\fill [black] (38.98,-44.08) -- (38.1,-44.43) -- (38.91,-45.02);
\draw [black] (38.977,-38.72) arc (234:-54:2.25);
\draw (40.3,-34.15) node [above] {$a'_0?$};
\fill [black] (41.62,-38.72) -- (42.5,-38.37) -- (41.69,-37.78);
\draw [black] (78.223,-31.88) arc (54:-234:2.25);
\draw (76.9,-36.45) node [below] {$\textit{timeout}?$};
\fill [black] (75.58,-31.88) -- (74.7,-32.23) -- (75.51,-32.82);
\draw [black] (75.577,-26.52) arc (234:-54:2.25);
\draw (76.9,-21.95) node [above] {$a'_0?,a'_1?$};
\fill [black] (78.22,-26.52) -- (79.1,-26.17) -- (78.29,-25.58);
\draw [black] (38.977,-14.42) arc (234:-54:2.25);
\draw (40.3,-9.85) node [above] {$a'_1?$};
\fill [black] (41.62,-14.42) -- (42.5,-14.07) -- (41.69,-13.48);
\draw [black] (41.623,-19.78) arc (54:-234:2.25);
\draw (40.3,-24.35) node [below] {$\textit{send}?$};
\fill [black] (38.98,-19.78) -- (38.1,-20.13) -- (38.91,-20.72);
\draw [black] (-2.6,-29.2) -- (0.8,-29.2);
\draw (-3.1,-29.2) node [left] {$\textit{Sndr}_H\mbox{ }:=$};
\fill [black] (0.8,-29.2) -- (0,-28.7) -- (0,-29.7);
\end{tikzpicture}
\end{center}
    \caption{A manually constructed ABP sender from~\cite{AlurTripakisSIGACT17}. We will refer to this sender as the {\it human-made} sender.}
    \label{fig:abp-sndr}
\end{figure}
\begin{figure}[!ht]
    \centering

\begin{center}
\begin{tikzpicture}[scale=\tiksize]
\tikzstyle{every node}+=[inner sep=0pt]
\draw [black] (3.8,-29.2) circle (3);
\draw (3.8,-29.2) node {$s_0$};
\draw [black] (22,-17.1) circle (3);
\draw (22,-17.1) node {$s_4$};
\draw [black] (58.6,-17.1) circle (3);
\draw (58.6,-17.1) node {$s_1$};
\draw [black] (76.9,-29.2) circle (3);
\draw (76.9,-29.2) node {$s_2$};
\draw [black] (22,-47.5) circle (3);
\draw (22,-47.5) node {$s_7$};
\draw [black] (40.3,-41.4) circle (3);
\draw (40.3,-41.4) node {$s_6$};
\draw [black] (58.6,-41.4) circle (3);
\draw (58.6,-41.4) node {$s_5$};
\draw [black] (40.3,-29.2) circle (3);
\draw (40.3,-29.2) node {$s_3$};
\draw [black] (2.477,-26.52) arc (234:-54:2.25);
\draw (3.8,-21.95) node [above] {$a'_0?,a'_1?$};
\fill [black] (5.12,-26.52) -- (6,-26.17) -- (5.19,-25.58);
\draw [black] (5.123,-31.88) arc (54:-234:2.25);
\draw (3.8,-36.45) node [below] {$\textit{timeout}?$};
\fill [black] (2.48,-31.88) -- (1.6,-32.23) -- (2.41,-32.82);
\draw [black] (6.3,-27.54) -- (19.5,-18.76);
\fill [black] (19.5,-18.76) -- (18.56,-18.79) -- (19.11,-19.62);
\draw (15.68,-23.65) node [below] {$\textit{send}?$};
\draw [black] (25,-17.1) -- (55.6,-17.1);
\fill [black] (55.6,-17.1) -- (54.8,-16.6) -- (54.8,-17.6);
\draw (40.3,-17.6) node [below] {$a'_0?$};
\draw [black] (61.1,-18.75) -- (74.4,-27.55);
\fill [black] (74.4,-27.55) -- (74.01,-26.69) -- (73.45,-27.52);
\draw (64.97,-23.65) node [below] {$\textit{done}!$};
\draw [black] (74.4,-30.86) -- (61.1,-39.74);
\fill [black] (61.1,-39.74) -- (62.04,-39.71) -- (61.48,-38.88);
\draw (63.36,-34.8) node [above] {$send?,a'_1?$};
\draw [black] (56.178,-43.159) arc (-60.33059:-119.66941:13.591);
\fill [black] (42.72,-43.16) -- (43.17,-43.99) -- (43.66,-43.12);
\draw (49.45,-45.44) node [below] {$p_1!$};
\draw [black] (37.45,-42.35) -- (24.85,-46.55);
\fill [black] (24.85,-46.55) -- (25.76,-46.77) -- (25.45,-45.82);
\draw (32.81,-45.02) node [below] {$a'_1?$};
\draw [black] (37.45,-40.45) -- (6.65,-30.15);
\fill [black] (6.65,-30.15) -- (7.25,-30.88) -- (7.56,-29.93);
\draw (24.53,-34.7) node [above] {$\textit{send}?$};
\draw [black] (19.88,-45.37) -- (5.92,-31.33);
\fill [black] (5.92,-31.33) -- (6.13,-32.25) -- (6.83,-31.54);
\draw (12.38,-39.83) node [left] {$\textit{done}!$};
\draw [black] (24.5,-18.75) -- (37.8,-27.55);
\fill [black] (37.8,-27.55) -- (37.41,-26.69) -- (36.85,-27.52);
\draw (35.15,-22.65) node [above] {$\textit{timeout}?$};
\draw [black] (37.319,-28.884) arc (-100.14368:-146.80185:20.981);
\fill [black] (23.46,-19.72) -- (23.48,-20.66) -- (24.31,-20.11);
\draw (27.81,-26.23) node [below] {$p_0!$};
\draw [black] (22.739,-14.205) arc (193.40804:-94.59196:2.25);
\draw (29.18,-11.48) node [above] {$\textit{send}?$};
\fill [black] (24.75,-15.93) -- (25.64,-16.23) -- (25.41,-15.25);
\draw [black] (19.202,-16.051) arc (277.19112:-10.80888:2.25);
\draw (15.5,-11.69) node [above] {$a'_1?$};
\fill [black] (21.13,-14.24) -- (21.53,-13.38) -- (20.53,-13.51);
\draw [black] (75.577,-26.52) arc (234:-54:2.25);
\draw (76.9,-21.95) node [above] {$a'_0?$};
\fill [black] (78.22,-26.52) -- (79.1,-26.17) -- (78.29,-25.58);
\draw [black] (78.223,-31.88) arc (54:-234:2.25);
\draw (76.9,-36.45) node [below] {$\textit{timeout}?$};
\fill [black] (75.58,-31.88) -- (74.7,-32.23) -- (75.51,-32.82);
\draw [black] (42.697,-39.606) arc (120.37251:59.62749:13.356);
\fill [black] (56.2,-39.61) -- (55.77,-38.77) -- (55.26,-39.63);
\draw (49.45,-37.27) node [above] {$a'_0?,timeout?$};
\draw [black] (-2.6,-29.2) -- (0.8,-29.2);
\draw (-3.1,-29.2) node [left] {$\textit{Sndr}_A\mbox{ }:=$};
\fill [black] (0.8,-29.2) -- (0,-28.7) -- (0,-29.7);
\end{tikzpicture}
\end{center}
    \caption{An automatically synthesized ABP sender from~\cite{AlurTripakisSIGACT17}. We will refer to this sender as the {\it algorithm-made} sender.}
    \label{fig:abp-sndr_pr}
\end{figure}

\begin{figure}[!ht]
    \centering
    \includegraphics[width=\textwidth]{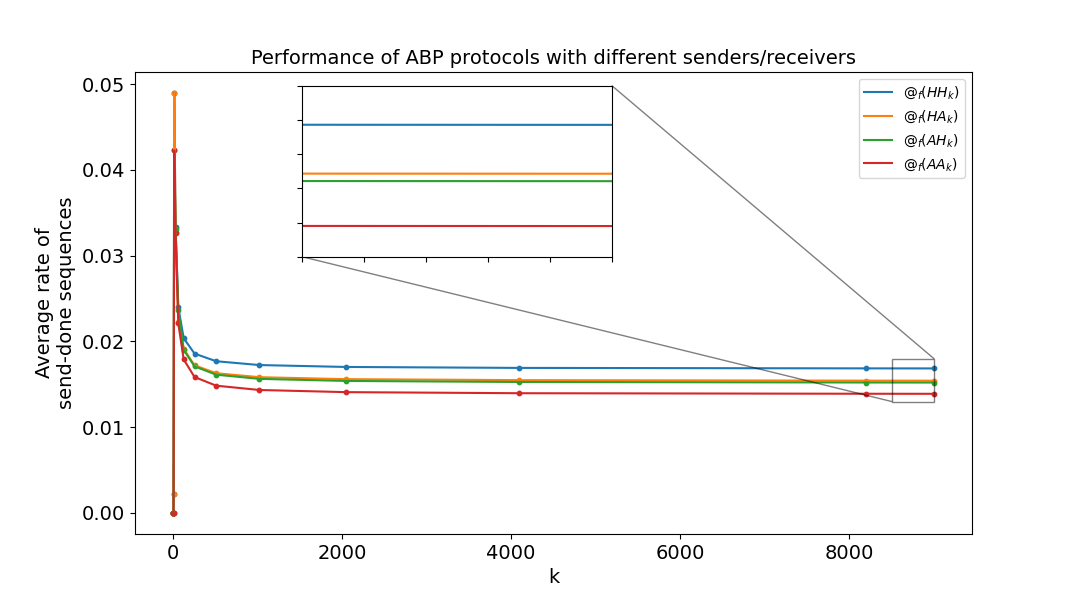}
    \caption{A graph of the $K$-approximation against $K$ for the ABP case study.}
    \label{fig:abp-approx}
\end{figure}

\end{document}